\documentclass[a4paper,11pt]{article}

\usepackage{amsmath,amssymb,graphicx,amsthm,dsfont}
\usepackage{color,soul}
\usepackage{xspace}
\usepackage{enumerate}
\usepackage{mathtools}
\usepackage{bbm}
\usepackage{comment}
\usepackage{paralist}
\usepackage{nicefrac}
\usepackage{booktabs}
\usepackage{multirow}
\usepackage{wrapfig}

\usepackage[margin=2.5cm]{geometry}

\usepackage{subcaption}
\usepackage[noend,ruled,linesnumbered]{algorithm2e} 
\usepackage{xifthen}

\usepackage{scrextend}

\usepackage{needspace}

\usepackage[numbers]{natbib}
\usepackage{nicefrac}

\usepackage[linecolor=green!70!black, backgroundcolor=green!10, bordercolor=black, textsize=small]{todonotes}

\usepackage[absolute]{textpos}  

\newcommand{\jensencite}{\cite{Jensen1906,GraRyz2000}}

\newcommand{\mydescrip}[1]{\smallskip

  \noindent\textcolor{darkgray}{\sffamily\bfseries\mathversion{bold}#1}}

\newcommand{\myemph}[1]{{\color{darkgreen!70!black}\emph{#1}}}

\newtheorem{theorem}{Theorem}
\newtheorem{corollary}{Corollary}
\newtheorem{example}{Example}
\newtheorem{lemma}{Lemma}
\theoremstyle{theorem}

\newtheorem{proposition}{Proposition}
\newtheorem{claim}{Claim}

\theoremstyle{remark}
\newtheorem*{claimproof}{Proof}

\newcommand{\ETH}{ETH}

\newcommand{\probname}[1]{\textsc{#1}}

\DeclareMathOperator*{\argmin}{arg\,min}

\usepackage{xcolor}

\definecolor{winered}{rgb}{0.6,0.1,0.1}
\definecolor{darkblue}{rgb}{0,0,0.4}
\definecolor{darkgreen}{rgb}{0.01,0.6,0.1}

\usepackage{tikz}
\usetikzlibrary{calc,arrows,shapes,shapes.geometric,matrix,decorations.pathreplacing,positioning,patterns,backgrounds,fit,topaths,mindmap}

\tikzset{matrixsc/.style={matrix of math nodes, ampersand replacement=\&, row sep=-7pt, column sep=-4pt}}

\makeatletter
\newcommand{\gettikzxy}[3]{%
  \tikz@scan@one@point\pgfutil@firstofone#1\relax
  \edef#2{\the\pgf@x}%
  \edef#3{\the\pgf@y}%
}
\makeatother

\newcommand{\prob}[6]{%
   \begin{quote}
    \begin{labeling}{#6}%
    \item[#1]
    \item[\emph{#2}]#3
    \item[\emph{#4}]#5
    \end{labeling}%
  \end{quote}%
}

\newcommand{\probdef}[3]{\prob{#1}{Input:}{#2}{Question:}{#3}{as}}

\newcommand{\myproofparagraph}[1]{\par\smallskip\noindent\fbox{\textit{#1}}}

\newcommand{\pHDClong}{\textsc{$p$-Norm Hamming Centroid}}
\newcommand{\pHDC}{\textsc{$p$-HDC}}

\newcommand{\pMinpHD}{\pHDC}

\newcommand{\CloStr}{\textsc{Closest String}}
\newcommand{\ConStr}{\textsc{Consensus String}}

\newcommand{\TCol}{{\normalfont\probname{$3$-Coloring}}}

\newcommand{\hd}[1]{\mathsf{d}^{#1}}

\newcommand{\hs}{\mathsf{hs}}
\newcommand{\col}{\mathsf{col}}
\newcommand{\zero}{\boldsymbol{0}}
\newcommand{\one}{\boldsymbol{1}}

\newcommand{\objectivel}[2]{\ensuremath{\sum_{s \in #1}\hd{p}(#2,s)}}

\newcommand{\objectiveroot}{\ensuremath{\|(\sol,S)\|_p}}
\newcommand{\sol}{\ensuremath{s^*}}

\newcommand{\distance}{(n+3m)\cdot \exactbound}
\newcommand{\gvalue}{(2^a+2^{a-b})\cdot  \hat{n}^p}
\newcommand{\guniqvalue}{(2^a+2^{a-b})\cdot  (\hat{n}+1)^p}
\newcommand{\kvalue}{2\hat{n}\cdot \sqrt[p]{2^{a-p}+2^{a-b-p}+2(n+3m)}}
\newcommand{\kvaluep}{(2^a+2^{a-b})\cdot \hat{n}^{p} + 2(n+3m)\cdot (2\hat{n})^{p}}
\newcommand{\functionforderivative}{\big((2^b+1)\cdot \hat{n}-y\big)^{p} +  2^{a-b}\cdot y^p}
\newcommand{\functionforderivativel}{\big((2^b+1)\cdot \hat{n}-y+1\big)^{p} +  2^{a-b}\cdot (y+1)^p}
\newcommand{\gadgetlb}{(2^{a}+2^{a-b})\cdot \hat{n}^{p}}

\newcommand{\exactbound}{2\cdot (2\hat{n})^p}

\usepackage[noend,ruled,linesnumbered]{algorithm2e}

\usepackage[sort&compress,nameinlink,capitalize]{cleveref}
\crefname{figure}{Figure}{Figures}
\crefname{algorithm}{Algorithm}{Algorithms}
\crefname{proposition}{Proposition}{Propositions}
\crefname{theorem}{Theorem}{Theorem}
\crefname{claim}{Claim}{Claims}
\crefname{lemma}{Lemma}{Lemmas}
\crefalias{AlgoLine}{line}

\usepackage{etoolbox} 

\newcommand{\appsymb}{$\star$}
\newcommand{\appref}[1]{{\hyperref[#1]{\appsymb}}}

\newcommand{\toappendix}[1]{%
    {#1}
}

\newcommand{\appendixsection}[1]{%
}
 \newcommand{\appendixproofwithstatement}[3]{%
    #3
 }

\newcommand{\mytitle}{On Computing Centroids According to the $p$-Norms of
  Hamming~Distance~Vectors}
\title{\mytitle}
\author{
  Jiehua Chen$^1$ \and Danny Hermelin$^2$ \and Manuel~Sorge$^1$\\
  {\small $^1$University of Warsaw, Warsaw, Poland}\\
  {\small $^1$\texttt{\{jiehua.chen2,manuel.sorge\}@gmail.com}}\\
  {\small $^2$Ben-Gurion University of the Negev, Beer Sheva, Israel}\\
  {\small $^2$\texttt{hermelin@bgu.ac.il}}%
}

\begin{document}
\thispagestyle{empty}
\maketitle


\begin{abstract}
In this paper we consider the $\pHDClong$ problem which asks to determine whether some given strings have a centroid with a bound on the $p$\nobreakdash-norm of its Hamming distances to the strings. Specifically, given a set~$S$ of strings and a real $k$, we consider the problem of determining whether there exists a string $\sol$ with $\big(\sum_{s \in S}\hd{p}(\sol,s)\big)^{1/p} \le k$, where $\hd{}(,)$ denotes the Hamming distance metric. This problem has important applications in data clustering and multi-winner committee elections, and is a generalization of the well-known polynomial-time solvable \ConStr{} $(p=1)$ problem, as well as the NP-hard \CloStr{} $(p=\infty)$ problem.

Our main result shows that the problem is NP-hard for all fixed rational $p > 1$, closing the gap for all rational values of $p$ between $1$ and $\infty$. Under standard complexity assumptions the reduction also implies that the problem has no $2^{o(n+m)}$-time or $2^{o(k^{\frac{p}{(p+1)}})}$-time algorithm, where $m$ denotes the number of input strings and $n$ denotes the length of each string, for any fixed $p > 1$. 
Both running time lower bounds are tight.
In particular, we provide a $2^{k^{\frac{p}{(p+1)}+\varepsilon}}$-time algorithm for each fixed~$\varepsilon > 0$.
In the last part of the paper, we complement our hardness result by presenting a fixed-parameter algorithm and a factor-$2$ approximation algorithm for the problem.
\end{abstract}

\newpage
\setcounter{page}{1}
\section{Introduction}
The \myemph{Hamming distance between two strings} of equal length is the number of positions at which the corresponding symbols in the strings differ. 
In other words, it measures the number of substitutions of symbols required to change one string into the other, 
or the number of errors that could have transformed one string into the other. This is perhaps the most fundamental string metric known, named after Richard Hamming who introduced the concept in 1950~\cite{Hamming1950}.

While Hamming distance has a variety of applications in a plethora of different domains, a common usage for it appears when clustering data of various sorts. Here, one typically wishes to cluster the data into groups that are centered around some centroid, where the notion of centroid varies from application to application. Two prominent examples in this context are:
\mydescrip{\ConStr{}}, where the centroid has a bound on the sum of its (Hamming) distance to all strings, and
  \mydescrip{\CloStr{}}, where the centroid has a bound on the maximum distance to all strings.

In functional analysis terms, these two problems can be formalized using the $p$-norms of the Hamming distance vectors associated with the clusters. That is, if $S \subseteq \{0,1\}^n$ is a cluster and $\sol \in \{0,1\}^n$ is its centroid, then the $p$-norm of the corresponding Hamming distance vector is defined~by

{\centering
  $\|(\sol,S)\|_p  \coloneqq \big(\sum_{s \in S}\hd{p}(\sol,s)\big)^{1/p},$
  \par}

\medskip

\noindent where $\hd{}(\sol,s) = |\{ i \colon \sol[i] \neq s[i], 1 \leq i \leq n \}|$ denotes the Hamming distance between~$\sol$ and~$s$. Using this notation, we can formulate \ConStr{} as the problem of finding a centroid~$\sol$ with a bound on $\|(\sol,S)\|_1$ for a given set~$S$ of strings, while \CloStr{} can be formulated as the problem of finding a centroid $\sol$ with a bound on~$\|(\sol,S)\|_\infty$. 

\begin{wrapfigure}[7]{r}{0pt}
  \begin{tikzpicture}
    \useasboundingbox[draw=none] (0.1,-.4) rectangle (8.4,1);
    \tikzstyle{every node}=[anchor=west]
    \matrix[matrixsc, row sep=.3pt] (M) {
      \& S: \& \\[-.2ex]
      \& 1111 \& 111 \\ 
      \& 1111 \& 000 \\ 
      \& 0000 \& 100 \\ 
      \& 0000 \& 010 \\
      \& 0000 \& 001 \\
    };
    
    \matrix[right = -1.7ex of M, matrix of math nodes, ampersand replacement=\&,  row sep=-1.2pt, column sep=0pt, 
    column 1/.style={nodes={text width=2.1ex}},
    column 2/.style={nodes={text width=0.7ex}},
    column 3/.style={nodes={text width=10ex}, anchor=base west},
    column 4/.style={nodes={text width=6ex, align=center}},
    column 5/.style={nodes={text width=6ex, align=center}},
    column 6/.style={nodes={text width=6ex, align=center}},
    row 1/.style={nodes={text height=2.3ex, align=center}},
    row 4/.style={nodes={text height=2.1ex, align=center}},
    row 3/.style={nodes={text height=2.1ex, align=center}},
    row 2/.style={nodes={text height=1.9ex, align=center}}] (MM) {
      {} \& {} \& {} \& \|\cdot\|_1 \& \|\cdot\|_2 \& \|\cdot\|_\infty\\
      \sol_1 \& \!=\! \&0000\, 000 \& |[fill=red!10]|14 \& \sqrt{68} \& 7\\
      \sol_2\&\!=\!\&0011\, 000\& 16 \& |[fill=red!10]|\sqrt{56} \& 5\\
      \sol_\infty\&\!=\!\&0011\, 001 \& 17 \&\sqrt{61} \& |[fill=red!10]|{4}\\
    };
     \draw (MM-1-1.north west) -- (MM-1-6.north east);
     \draw (MM-1-1.north west) -- (MM-4-1.south west);

     \gettikzxy{(MM-4-1.south west)}{\xx}{\yy};
     \gettikzxy{(MM-4-6.south east)}{\zz}{\ww};

     \draw (MM-1-6.north east) -- (\zz,\yy);

     \draw (\xx,\yy) -- (\zz,\yy);


     \draw[] ($(MM-1-1.west)!.5!(MM-2-1.west)$) -- ($(MM-1-6.east)!.5!(MM-2-6.east)$);
     \draw[] ($(MM-1-3.north)!.5!(MM-1-4.north)$) -- ({$(MM-4-3.south)!.5!(MM-4-4.south)$} |- MM-4-1.south west);

     \draw[] ($(MM-1-1.north west)$) -- ($(MM-1-3.south)!.5!(MM-1-4.south)$);
     
     \node at ($(MM-1-3)+(0.08,0.08)$) {$p$};
     \node at ($(MM-1-2)+(-.5,-0.2)$) {\small centroid};
   \end{tikzpicture}
   \end{wrapfigure}
The following cluster~$S$ with $5$ strings, each of length $7$, shows that for different~$p$, we indeed obtain different optimal centroids.
For each $p\in \{1,2,\infty\}$, string~$\sol_p$ is an optimal $p$-norm centroid but it is not an optimal $q$-norm centroid, where $q\! \in\! \{1, 2, \infty\} \!\setminus\! \{p\}$.
Moreover, one can verify that $\sol_2$ is the only optimal $2$-norm centroid and no optimal $\infty$-norm centroid is an optimal $2$-norm centroid.

\looseness=-1
The notion of $p$-norms for distance vectors is very common in many different research fields~\cite{MonAffHarBar1982,LovMorWes1988,Gonin1989,NesNem1994,GhiRom1996,AzEpRiWo2004,KloBreSonLasMueZie2009,BelSolCal2016,faliszewski_multiwinner_2017,sivarajan_generalization_2018}.
In cluster analysis of data mining and machine learning, 
one main goal is to partition $m$~observations (i.e., $m$~real vectors of the same dimension) into $K$~groups so that the sum of ``discrepancies'' between each observation and its nearest center is minimized.
Here, two highly prominent clustering methods are \emph{$K$-means}~\cite{Macqueen67somemethods} and \emph{$K$-medians}~\cite{JaiDub1988,BraManStr1996} clustering, each using a slightly different notion of discrepancy measure.
The first method aims to minimize the \emph{sum of squared Euclidean distances} between each observation and the ``mean'' of its respective group. In other words, it minimizes the squared $2$-norm of the Euclidean-distance vector.
$K$-medians, on the other hand, uses the $1$-norm instead of the squared $2$-norm to define the discrepancy to the mean.
Thus, instead of calculating the mean for each group to determine its centroid, one calculates the median.

\looseness=-1
In committee elections from social choice
theory~\cite{faliszewski_multiwinner_2017,sivarajan_generalization_2018,PascRzaSko2018,FalSkoSliTal2019},
the $p$-norm is used to analyze how well a possible committee
represents the voter's choices. In a fundamental approval-based procedure to
select a $t$-person committee
from $n$ candidates, each voter either approves or disapproves each of the candidates,
which can be expressed as a binary string of length~$n$.
An optimal committee is a length-$n$ binary string containing exactly $t$ ones and which minimizes the $p$-norm of the vector of the Hamming distances to each voter's preference string~\cite{sivarajan_generalization_2018}.



\paragraph{Problem definition, notations, and conventions.}
Since the Hamming distance is frequently used in various applications, e.g., in computational biology~\cite{Pev2000}, information theory, coding theory and cryptography~\cite{Hamming1950,CoHoLiLo1997,Roth2006}, in social choice~\cite{Kilgour2010,AmaBarLanMarRie2015} and since the notion of $p$-norm is very prominent in clustering tools~\cite{shier_optimal_1983,brandeau_parametric_1988,LovMorWes1988,ZengSoZoubir2013} and preference aggregation rules~\cite{AmaBarLanMarRie2015,brams_minimax_2007,PascRzaSko2018}, where often $p = 1, 2, \infty$ but also other values of $p$ are used, it is natural to consider computational problems associated with the $p$-norm of the Hamming distance metric.
This is the main purpose of this paper. Specifically, we consider the following problem:
\probdef{$\pHDClong$ (\pHDC)}
{A set~$S$ of strings~$s_1, \ldots, s_m \in \{0, 1\}^n$ and a real~$k$.}
{Is there a string $\sol \in \{0, 1\}^n$ such that $\objectiveroot \leq k$?}%
Throughout, we will call a string~$\sol$ as above a \myemph{solution}. Note that there is nothing special about using the binary alphabet in the definition above, but for ease of presentation we use it throughout the paper.
When $p=1$, our $\pHDC$ problem is precisely the $\ConStr$ problem, and when $p=\infty$ it becomes the $\CloStr$ problem.

In the following, we list some notation and conventions that we use.
By \myemph{$p$-distance} we mean the $p^{\text{th}}$-power of the Hamming distance.
For each natural number $t$ by $[t]$ we denote the set $\{1, 2, \ldots, t\}$. Unless stated otherwise, by \myemph{strings} we mean binary strings over alphabet~$\{0, 1\}$.
Given a string~$s$, we use $|s|$ to denote the length of this string.
For two binary strings~$s$ and $s'$, let $s\circ s'$ denote the concatenation of $s$ and $s'$. By $s[j]$ we denote the $j$th value or the value in the $j^{\text{th}}$ character of string~$s$. 
By $\overline{s} = (1 - s[j])_{j \in [|s|]}$ we denote the complement of the (binary) string~$s$.
Given two integers~$j,j'\in \{1,2,\cdots, |s|\}$ with $j \le j'$, 
we write $s|^{j'}_{j}$ for the substring $s[j]s[j+1]\cdots s[j']$.
Given a number~$\ell$, 
we use $\zero_\ell$ and $\one_\ell$ to denote the length-$\ell$ all-zero string and the length-$\ell$ all-one string, respectively.

\paragraph{Our contributions.}
Our main result is a tight running time bound on the \pHDC{} problem for all fixed rationals~$p > 1$. Specifically, we show that the problem is NP-hard and can be solved in $2^{k^{p/(p+1)+\varepsilon}}\cdot |I|^{O(1)}$ time for arbitrary small $\varepsilon > 0$ where $|I|$ denotes the size of the instance, but cannot be solved in $2^{o(k^{p/(p+1)})}$ time unless the Exponential Time Hypothesis (ETH)~\cite{CyFoKoLoMaPiPiSa2015} fails. The lower bounds are given in \cref{thm:mphd-np-h,prop:nph-distinct} and the upper bound in \cref{thm:subexp}. While the upper bound in this result is not very difficult, the lower bound uses an intricate construction and some delicate arguments to prove its correctness.
In particular, the construction extensively utilizes the fact that since $p>1$, the $p$-norm of Hamming distances is convex and always admits a second derivative.
We believe that this kind of technique is of interest on its own.
As another consequence of the hardness construction, we also obtain a $2^{o(n+m)}$ running time lower bound assuming ETH, which gives evidence that the trivial brute-force $2^n\cdot |I|$-time algorithm for the problem cannot be substantially improved. Moreover, the lower bounds also hold when we constrain the solution string to have a prescribed number of ones. That is, we also show hardness for the committee election problem mentioned above (\cref{cor:committee-nph}).

In the final part of the paper we present two more algorithms for \pHDC.
First, we provide an $m^{O(m^2)}\cdot |I|^{O(1)}$ time algorithm (see \cref{thm:fpt-m}), by first formulating the problem as a so-called Combinatorial $n$-fold Integer Program, and then applying the algorithm developed by \citet{KnoKouMni2017}.
Second, we show that the problem can be approximated in polynomial time within a factor of 2, using an extension of the well known 2-approximation algorithm for \textsc{Closest String} (see \cref{prop:factor-2}).


\paragraph{Related work.}
The NP-complete \CloStr~\cite{FraLit1997,LaLiMaWaZh2003} problem (aka.\ \textsc{Minimum Radius}
) is a special case of \pMinpHD{} with $p=\infty$.
It seems, however, difficult to adapt this hardness reduction to achieve our hardness results for every fixed rational $p$ (see also the beginning of \cref{sec:nphs} for some more discussion).
\CloStr{} has been studied extensively under the lens of parameterized complexity and approximation algorithmics.
The first fixed-parameter algorithm for parameter $k$, the maximum Hamming distance bound, was given by \citet{GraNieRoss2003}, runs in $O(k^k\cdot km + mn)$ time where $m$ and $n$ denote the number and the length of input strings, respectively. This algorithm works for arbitrary alphabet~$\Sigma$.
For small alphabets~$\Sigma$, there are algorithms with $O(mn + n\cdot |\Sigma|^{O(k)})$~running time~\cite{MaSun2009,chen_three-string_2012}. Both types of running time are tight under the ETH~\cite[Theorem 14.17]{CyFoKoLoMaPiPiSa2015}. For arbitrary alphabet~$\Sigma$, \citet{KnoKouMni2017} gave an algorithm with $m^{O(m^2)}\cdot \log{n}$ running time based on so-called $\mathsf{n}$-fold integer programming.
As for approximability, \CloStr{} admits a PTAS with running time~$O(n^{O(\epsilon^{-2})})$~\cite{MaSun2009} but no EPTAS unless FPT${}={}$W[1]~\cite{CygLokPilPilSau2016}.

Our problem falls into the general framework of convex optimization with binary variables.
If a solution is allowed to have fractional values, then the underlying convex optimization can be solved in polynomial time for each fixed value~$p\le 2$~\cite[Chapter 6.3.2]{NesNem1994}.

For $p=2$, maximizing (instead of minimizing) the $p$-norm reduces to \textsc{Mirkin Minimization} in consensus clustering with input and output restricted to two-clusters, which was shown to be NP-hard~\cite{DoGuKoWe2014} under Turing reductions.
Recently, \citet{CheHerSor2018mirkin} showed that the simple $2^n$-time algorithm by brute-force searching all possible outcome solutions is essentially tight under ETH.
They also provided some efficient algorithms and showed that the problem admits an FPTAS using a simple rounding technique.



\section{NP-hardness for the \boldmath$p$-norm of Hamming distance vectors}\label{sec:nphs}
\appendixsection{sec:nphs}
\looseness=-1
We now show that \pHDC{} is NP-hard for each fixed
rational number~$p > 1$ (\cref{thm:mphd-np-h,prop:nph-distinct}) and
that algorithms with running time $2^{o(n + m)}$ or
$2^{o(k^{\nicefrac{p}{(p+1)}})}$ would contradict the \ETH. We reduce
from the NP-hard \TCol{} problem~\cite{GJ79} in
which, given an undirected graph~$G=(V,E)$, we ask whether there is a
\myemph{proper vertex coloring}~$\col\colon V\to \{0,1,2\}$, that is, no two
\emph{adjacent} vertices receive the same color. 

The first challenge we need to overcome when designing the reduction
is to produce some regularity in the solution string:
Given~$\hat{n} \in \mathds{N}$, in \cref{lem:gadget-n}, we show how to
construct a set of strings to enforce a solution string to have exactly $\hat{n}$ ones which only occur in the columns of some specific range.
This allows us later on to build gadgets that have several discrete
states. Indeed, after controlling the overall number of ones in the
solution in this way, we can allocate three columns (one for each color) for each vertex~$v$
in~$G$ and build a gadget for~$v$ such that this
gadget induces minimum $p$-distance to the solution if and only if
there is exactly $1$ one in the solution in the columns allocated for~$v$.
This column determines the color for~$v$. Then, for each edge, we will
introduce an edge gadget consisting of six strings which induce
minimum $p$-distance in the solution if and only if they are ``covered'' by
the ones in the solution exactly twice, corresponding to different
colors.

In general, the design of gadgets for \pHDC\ is quite different from
the known NP-hard case \CloStr\
($p = \infty$)~\cite{FraLit1997,LaLiMaWaZh2003}: In \CloStr\ every
optimal solution~$\sol$ must regard the ``worst'' possible input
string while in our case $\sol$ can escape such constraints by
distributing some of its Hamming distance from the ``worst'' to other strings.

In the remainder of this section, let $a$ and $b$ be two fixed
integers such that $a$ and $b$ are coprime
, $a>b$, and $p = a/b > 1$. To better capture the Hamming distance, we
introduce the notion of the \myemph{Hamming set} of two strings~$s$ and
$s'$ of equal length~$n$, which consists of the indices of the columns at
which both strings differ:
$\hs(s,s')=\{j\in [n] \mid s[j]\neq s'[j]\}$.

\looseness=-1 As mentioned, we first show how to construct a set of
strings to enforce some structure on the optimal solution, that is, a
binary string with minimum sum of the $p$-distances.
\newcommand{\gadgetlemma}{%
    Let $p\!>\!1$ be a fixed rational number, and $a$ and $b$ be two coprime fixed integers with $p\!=\!a/b$.
  Let $S$ consist of one string~$\one_{(2^b+1)\hat{n}}\!\circ\! \zero_{\hat{n}}$ and
  $2^{a-b}$ copies of string~$\zero_{(2^b+2)\hat{n}}$, where $\hat{n}$ is a positive integer. 
  For each string~$\sol \in \{0,1\}^{(2^b+2)\hat{n}}$, the following holds.
  \begin{compactenum}[(1)]
    \item\label{lem-gadget:=} 
    If $\hd{}(\sol,\zero_{(2^b+2)\hat{n}}) =\hat{n}$ and $\hs(\sol,\zero_{(2^b+2)\hat{n}})\subseteq [(2^b+1)\hat{n}]$, 
    then
    $\|(\sol, S)\|^p_p=\gadgetlb$.
    \item\label{lem-gadget:>} If $\hd{}(\sol,\zero_{(2^b+2)\hat{n}}) \neq \hat{n}$ or $\hs(\sol,\zero_{(2^b+2)\hat{n}})\nsubseteq [(2^b+1)\hat{n}]$,
    then
    $\|(\sol,S)\|_p^p > \gadgetlb$.
  \end{compactenum}%
}
\begin{lemma}[\appsymb\footnote{Proofs for results marked by \appsymb\ are deferred to an appendix.}]\label[lemma]{lem:gadget-n}
  \gadgetlemma
\end{lemma}
\appendixproofwithstatement{lem:gadget-n}{\gadgetlemma}{
\begin{proof}
  The first statement is straightforward to see by a simple calculation.
  
  We now prove the second statement.
  Let $y$ equal the number of ones in~$\sol$ in the first $(2^b+1)\cdot \hat{n}$ columns. Then, $\|(\sol,S)\|^p_p \ge \functionforderivative$.
  We define a function~$f \colon [0,(2^b+1)\cdot \hat{n}] \to \mathds{Z}$ with $f(y)\coloneqq \functionforderivative$,
  and show that~$f$ attains its sole minimum over $[0, (2^b+1)\cdot \hat{n}]$ at $y=\hat{n}$.
  Note that~$f$ is a lower bound on the sum of $p$-distances from~$\sol$ to~$S$.
  Furthermore, if $\sol$ has a one in the last $\hat{n}$ columns, then the sum of $p$-distances of $\sol$ is strictly larger than~$f(y)$ because each string from $S$ has only zeros in the last $\hat{n}$~columns.
  The first derivative of $f$ with respect to~$y$ is 
\begin{align}
  \frac{d f}{d y} &= -p \cdot \big((2^b+1) \hat{n}-y\big)^{p-1}+ p\cdot 2^{a-b}y^{p-1}\nonumber\\
  & = p\cdot \Big( 2^{a-b}y^{p-1} -  \big((2^b+1) \hat{n}-y\big)^{p-1} \Big).\label{eq:lb-first-derivative}
\end{align}
Now, observe that the first derivative of $f$ is zero when
the second multiplicand \eqref{eq:lb-first-derivative} is zero, 
because $p>1$,
that is, when
\begin{align} 2^{a-b}y^{p-1} -  \big((2^b+1) \hat{n}-y\big)^{p-1}= 0.\label{eq:lb-first-derivative2} \end{align}
Again, since $p>1$ we can infer that \eqref{eq:lb-first-derivative2} holds when
$2^b\cdot y=(2^{b}+1)\hat{n}-y$.
This is the case only when $y = \hat{n}$.

The second derivative of $f$ respect to $y$ is
\begin{align*}
  \frac{d^2f}{dy^2} & = p \cdot (p-1) \cdot \Big( 2^{a-b} \cdot y^{p-2}~+~\big((2^b+1)\hat{n}-y\big)^{p-2}\Big),
\end{align*}
which is positive at $y=\hat{n}$
(recall that $p > 1$).
Hence, indeed, the sole
minimum of $f(y)$ over $[0, (2^b + 1)\cdot \hat{n}]$ is attained at $y = \hat{n}$ with
$f(\hat{n})=\gadgetlb$.

To summarize, if $\hs(\sol,\zero_{(2^b+2)\hat{n}})\subseteq [(2^b+1)\hat{n}]$ but $\hd{}(\sol,\zero_{(2^b+2)\hat{n}}) \neq \hat{n}$, then $\|(\sol, S)\|^p_p = f(y) > f(\hat{n})$.
If $\hs(\sol,\zero_{(2^b+2)\hat{n}})\nsubseteq [(2^b+1)\hat{n}]$,
then $\sol$ has at least $1$ one in the last $\hat{n}$~columns.
Since the last $\hat{n}$~columns of each string from $S$ are all zeros,
it follows that $\|(\sol, S)\|^p_p \ge \functionforderivativel > f(y) \ge f(\hat{n})$.
\end{proof}
}

\noindent
To show \cref{lem:gadget-n} we crucially use the fact that $p > 1$.
In contrast, if $p = 1$, then taking the majority value in each column
yields an optimal solution, and thus it is impossible to force every optimal solution
to have a certain number of ones without at the same time specifying in which precise columns these ones should occur.

In the reduction we make heavy use of specific pairs of strings whose Hamming
distances to an arbitrary string always sum up to some lower bound. They will enforce
local structure in some columns of the solution, while being somewhat
immune to changes elsewhere. As a tool in the reduction we derive the
following lower bound on the $p$-distance of an arbitrary string to a
pair of strings which are quite far from each other, in terms of Hamming
distances.

\newcommand{\completementlowerbound}{%
 Let $s_1$ and $s_2$ be two strings of the same length~$R$ such that the Hamming distance between $s_1$ and $s_2$ is $\hd{}(s_1,s_2) = 2L$.
  For each rational~$p>1$ and each length-$R$ string~$\hat{s}$ the following holds.
  \begin{inparaenum}[(1)]
    \item\label{lem-complement:1} $\hd{p}(\hat{s}, s_1)+\hd{p}(\hat{s}, s_2) \ge 2\cdot L^p$.
    \item\label{lem-complement:2} If $\hd{}(\hat{s}, s_1) = \hd{}(\hat{s},s_2)= L$, then $\hd{p}(\hat{s}, s_1)+\hd{p}(\hat{s}, s_2) = 2\cdot L^p$. 
    \item\label{lem-complement:3} If $\hd{}(\hat{s}, s_1) \neq L $ or $\hd{}(\hat{s},s_2)\neq  L$, 
    then $\hd{p}(\hat{s}, s_1)+\hd{p}(\hat{s}, s_2) > 2\cdot L^p$. 
    \end{inparaenum}
}
\begin{lemma}[\appsymb]\label[lemma]{lem:complement}
  \completementlowerbound
\end{lemma}
  \appendixproofwithstatement{lem:complement}{\completementlowerbound}{
    \begin{proof}
      To simplify the notation,
      we define a convex function~$f\colon \mathds{R}^+ \cup \{0\} \to \mathds{R}$
      with $f(x)=x^p$; recall that $p >1$ so $f$ is indeed convex.
      
      To show the first statement, we will use Jensen's inequality~\jensencite\ on convex functions and the fact that Hamming distances satisfy the triangle inequality, as follows.
      \begin{align}
        \hd{p}(\hat{s},s_1) +  \hd{p}(\hat{s},s_2) & = f(\hd{}(\hat{s},s_1)) +  f(\hd{}(\hat{s},s_2))\label{lemma2-eq:def-f}\\
                                                   & \ge 2 \cdot f\big(\frac{\hd{}(\hat{s},s_1)+  \hd{}(\hat{s},s_2)}{2}\big)\label{lemma2-eq:Jensens}\\
                                                   & \ge 2 \cdot f(\frac{\hd{}(s_1,s_2)}{2})\label{lemma2-eq:Triangle}\\
        & = 2 \cdot L^p.
      \end{align}
      The first equation, (\ref{lemma2-eq:def-f}), follows by our definition of $f$, inequality (\ref{lemma2-eq:Jensens}) follows by Jensen's inequality~\jensencite, inequality (\ref{lemma2-eq:Triangle}) follows from the fact that Hamming distances satisfy the triangle inequality, while the last equation follows from our assumption on $\hd{}(s_1,s_2)$.

      The second statement can be verified by a straightforward calculation.

      The last statement holds by utilizing the fact that Jensen's inequality holds with equality if and only if (in the above instantiation) $\hd{}(\hat{s},s_1)=\hd{}(\hat{s},s_2)$.
      By assumption, this is the case only when $\hd{}(\hat{s},s_1)=\hd{}(\hat{s},s_2)=L$.
      \end{proof}
}

\noindent Using \cref{lem:gadget-n,lem:complement}, we can show NP-hardness of \pMinpHD{} for each fixed rational~$p> 1$.
For better readability, we will first show hardness for the case with multiple identical strings~%
(\cref{thm:mphd-np-h}) and then extend the construction to also include the case where no two strings are the same (\cref{prop:nph-distinct}).

\begin{theorem}\label[theorem]{thm:mphd-np-h}
  For each fixed rational number~$p > 1$, 
  \pHDC{} (with possibly multiple identical strings) 
  is NP-hard.
\end{theorem}

\begin{proof}
  First of all, let $a$ and $b$ be two fixed coprime integers such that $p=a/b$.
  To show the hardness result, we reduce from the NP-hard \TCol{} problem~\cite{GJ79} defined above.
  %
  Let $G=(V, E)$ be an instance of \TCol{}. Let $n$ be the number of vertices in~$G$ and $m$ the number of edges. Denote $V=\{v_1,v_2,\ldots,v_n\}$ and $E=\{e_1,e_2,\ldots, e_m\}$.

  \looseness=-1
  \myproofparagraph{Construction.}
  We introduce three groups of strings of
  length~$(2^b+2)\cdot \hat{n}$ each, where $\hat{n}=n+m$.
  The first group ensures that each optimal solution string must have exactly
  $\hat{n}$ ones which appear in the first $3\hat{n}$ columns (using \cref{lem:gadget-n}), the
  second group ensures that an optimal solution enforces that each vertex has exactly one of the three colors, and the third group, combined with the
  second group, ensures that no two adjacent vertices obtain the same
  color.

   \mydescrip{Group 1.} Construct one string~$\one_{(2^b+1) \hat{n}}\circ \zero_{\hat{n}}$ and $2^{a-b}$ copies of the same string~$\zero_{(2^b+2)\hat{n}}$.  

   \mydescrip{Group 2.} This group consists of one pair of strings for
   each vertex. Each pair consists of two strings which are mostly
   complements to each other. This ensures that the Hamming distance
   to the solution induced by a pair is somewhat homogeneous,
   regardless where exactly the ones in the solution occur. However,
   in each pair there are three columns, corresponding to the vertex,
   which will skew the pairs of Hamming distances in a way to induce
   minimum $p$-distances only if the solution has exactly $1$ one in these three columns.



   Formally, for each vertex~$v_i\in V$, 
   let $u_i$ be a string of length $3\hat{n}$ which has exactly $3$ ones in the columns $3i-2, 3i-1,3i$,
   and let $\overline{u}_i$ be the complement of $u_i$. Deriving from $u_i$, we construct two \myemph{vertex strings}
    $s_i$ and $r_i$ with $s_i\coloneqq u_i\circ \zero_{(2^b-2)\hat{n}} \circ 0 \circ \one_{\hat{n}-1}$ and $r_i \coloneqq \overline{u}_i \circ \zero_{(2^b-2)\hat{n}} \circ 1 \circ \zero_{\hat{n}-1}$.
    Note that both strings~$s_i$ and $r_i$ have all zeros in the columns~$\{3\hat{n},\ldots, (2^b+1)\hat{n}\}$ such that $\hd{}(s_i,r_i)=4\hat{n}$.

    \smallskip

    \noindent For an illustration, 
    the strings~$s_2$ and $r_2$, which correspond to the vertex~$v_2$, are as follows:

    
    
      $ s_2 =   {\color{winered}000\, 111\,} \circ {\color{winered}\zero_{3\hat{n}-6}} \circ \zero_{(2^b-2)\hat{n}} \circ {\color{winered} 0} \circ {\color{winered}\one_{\hat{n}-1}}, $~~
      $  r_2  =  {\color{winered}111\, 000\,} \circ  {\color{winered}\one_{3\hat{n}-6}} \circ\zero_{(2^b-2)\hat{n}} \circ {\color{winered}1} \circ {\color{winered}\zero_{\hat{n}-1}}. $

      \mydescrip{Group 3.} \looseness=-1 We now use three pairs of
      strings for each edge to ensure relatively homogeneous
      distributions of Hamming distances to the solution and then skew
      them. This time, we aim to skew distances to the solution so
      that their corresponding $p$-distances are minimum only if the
      solution distributes exactly three ones (corresponding to the
      colors) over three special regions: two corresponding to the
      endpoints of the edge and one extra dummy region.

    Formally, for each edge~$e_j\in E$ 
    let $e^{(0)}_j, e^{(1)}_j$, and $e^{(2)}_j$ denote three strings, each of length~$3\hat{n}$, that ensure that the edge and both of its endpoints each have a distinct color:
    
    \allowdisplaybreaks
    $
    \forall \ell  \in \{1,2,\cdots, \hat{n}\}\colon
      e^{(0)}_{j}[3\ell-2,3\ell-1,3\ell] \coloneqq
    \begin{cases}
        100, & 1\le \ell \le n \text{ with } v_\ell \in e_j \text{, or } \ell=j+n,\\
        000, & \text{otherwise.}
      \end{cases}$

     \quad  \qquad   \qquad \qquad   ~~~~$ e^{(1)}_{j}[3\ell-2,3\ell-1,3\ell] \coloneqq
   \begin{cases}
        010, & 1\le \ell \le n \text{ with } v_\ell \in e_j \text{, or } \ell=j+n,\\
        000, & \text{otherwise.}
      \end{cases}$
      
     \quad  \qquad   \qquad \qquad ~~~~$ e^{(2)}_{j}[3\ell-2,3\ell-1,3\ell] \coloneqq
     \begin{cases}
        001, & 1\le \ell \le n \text{ with } v_\ell \in e_j \text{, or } \ell=j+n,\\
        000, & \text{otherwise.}
      \end{cases}
      $



    \noindent Now, we construct the following six \myemph{edge strings} for edge~$e_j$:
  \begin{align*}
    \forall z \in \{0,1,2\}\colon t^{(z)}_j \coloneqq {\color{winered}e^{(z)}_j} \circ \zero_{(2^b-2)\hat{n}}\circ {\color{winered}0} \circ {\color{winered}\one_{\hat{n}-1}} \text{~~ and~~ }
    w^{(z)}_j \coloneqq {\color{winered}\overline{e}^{(z)}_j} \circ \zero_{(2^b-2)\hat{n}}\circ {\color{winered}1} \circ {\color{winered}\zero_{\hat{n}-1}}.
  \end{align*}
  Just as for group~2, the two strings~$t^{(z)}_j$ and $w^{(z)}_j$ have all zeros in the columns~$\{3\hat{n},\ldots, (2^b+1)\hat{n}\}$ such that $\hd{}(t^{(z)}_j, w_j^{(z)})=4\hat{n}$.
%
  \noindent For an example, assume that $a=3$, $b=2$, $n=3$, and $m=2$, and there is an edge of the form~$e_2=\{v_1,v_3\}$.
  Then, the two triples of strings that we construct for $e_2$ have each length $(2^b+2)(n+m)=30$ and are

\noindent  \begin{tabular}{l@{}c@{}ll@{}c@{}l}
      $t^{(0)}_j$ &$~=~$& ${\color{winered}100\,000\,100\,000\,100}\,0000000000\,{\color{winered}01111}$, &     $w^{(0)}_j$ &$~=~$& ${\color{winered}011\,111\,011\,111\,011}\,0000000000\,{\color{winered}10000}$,\\
      $t^{(1)}_j$ &$~=~$& ${\color{winered}010\,000\,010\,000\,010}\,0000000000\,{\color{winered}01111}$,  &    $w^{(1)}_j$ &$~=~$& ${\color{winered}101\,111\,101\,111\,101}\,0000000000\,{\color{winered}10000}$, \\
    $t^{(2)}_j$ &$~=~$& ${\color{winered}001\,000\,001\,000\,001}\,0000000000\,{\color{winered}01111}$, &
                $w^{(2)}_j$ &$~=~$& ${\color{winered}110\,111\,110\,111\,110}\,0000000000\,{\color{winered}10000}$.
  \end{tabular}

  \smallskip

 \looseness=-1
  \noindent  Summarizing, the instance~$I'$ of \pHDC{}
  consists of the following strings, each of length~$(2^{b}+2)\hat{n}=(2^b+2)(n+m)$:
  \begin{compactenum}[(1)]
    \item Add the $2^{a-b}+1$ strings in group 1 to $I'$.
    \item For each vertex~$v_i\in V$, add the vertex strings~$s_i$ and $r_i$ to $I'$.
    \item For each edge~$e_j\in E$, add two triples~$t^{(0)}_j$, $t^{(1)}_j$, $t^{(2)}_j$ and $w^{(0)}_j$, $w^{(1)}_j$, $w^{(2)}_j$ to $I'$.
    \end{compactenum}
    See Figure~1 for an illustration.

       \begin{figure}[t!]
  \tikzstyle{vertex} = [draw=gray, fill=gray, circle, inner sep=2pt]
  \begin{tikzpicture}
    \def \scalee {1}
  \node[]  at (-1*\scalee,1*\scalee) (G) {$G:$};
  \foreach \i / \j / \n / \c / \s in {0/0/3/blue/2, 1.4/0/4/red/0, 0/1.4/1/red/0, 1.4/1.4/2/green/1} {
    \node[vertex,draw=\c,fill=\c!70] at (\i*\scalee,\j*\scalee) (v\n) {{\color{white}$\boldsymbol\s$}};
  }

  \foreach \n / \o in {1/above, 2/above, 3/below, 4/below} {
    \node[\o = 0pt of v\n] (n\n) {$v_\n$};
  }

  \foreach \s / \t / \n in {1/2/1,2/3/5,3/4/3,2/4/2,1/3/4} {
    \path[draw,thick] (v\s) -- node[midway,fill=white, inner sep=1pt] {$e_\n$} (v\t);
  }

  \node[text width=5.7cm,anchor=south west,  below = 1.9cm of G, xshift=2cm] {Figure~1:~Illustration of the reduction used in \cref{thm:mphd-np-h}. The left figure depicts a graph~$G$ 
    that admits a proper vertex coloring~$\col$ (see the labels on the vertices). For instance, vertex~$v_1$ has color~$0$, i.e., $col(v_1)=0$.
    The right figure shows the crucial part of an instance of \pHDC{} with $p=2$ (i.e., $a=2$ and $b=1$) that we will construct according to the proof for \cref{thm:mphd-np-h}. In every pair of constructed strings we only show the first one.
    A solution string~$\sol$ corresponding to the coloring~$\col$ is depicted at the bottom of the right figure.};

    \matrix[matrixsc, right = 12.5cm of v2, anchor=north east, yshift=.2cm] (M) {
      {}    \& 111 \& 111 \& 111 \& 111 \& 111 \& 111 \& 111 \& 111 \& 111 \& 0000 \& 00000\\
       {}   \& 000 \& 000 \& 000 \& 000 \& 000 \& 000 \& 000 \& 000 \& 000 \& 0000 \& 00000\\
       {}   \& 000 \& 000 \& 000 \& 000 \& 000 \& 000 \& 000 \& 000 \& 000 \& 0000 \& 00000\\[.9ex]
      s_1:\& |[fill=gray!15,inner sep=1pt]|111 \& 000 \& 000 \& 000 \& 000 \& 000 \& 000 \& 000 \& 000  \& 0111 \& 11111\\
      s_2:\&  000 \&  |[fill=gray!15,inner sep=1pt]|111 \& 000 \& 000 \& 000 \& 000 \& 000 \& 000 \& 000  \& 0111 \& 11111\\
      s_3:\&  000 \& 000 \&  |[fill=gray!15,inner sep=1pt]|111 \& 000 \& 000 \& 000 \& 000 \& 000 \& 000  \& 0111 \& 11111\\
      s_4:\&  000 \& 000 \& 000 \&  |[fill=gray!15,inner sep=1pt]|111 \& 000 \& 000 \& 000 \& 000 \& 000  \& 0111 \& 11111\\[0.4ex]
      t_1^{(0)}: \&  |[fill=gray!15,inner sep=1pt]|100 \&  |[fill=gray!15,inner sep=1pt]|100 \& 000 \& 000 \&  |[fill=gray!15,inner sep=1pt]|100 \& 000 \& 000 \& 000 \& 000 \& 0111 \& 11111 \\
      t_1^{(1)}: \&  |[fill=gray!15,inner sep=1pt]|010 \&  |[fill=gray!15,inner sep=1pt]|010 \& 000 \& 000 \&  |[fill=gray!15,inner sep=1pt]|010 \& 000 \& 000 \& 000 \& 000 \& 0111 \& 11111 \\
      t_1^{(2)}: \&  |[fill=gray!15,inner sep=1pt]|001 \&  |[fill=gray!15,inner sep=1pt]|001 \& 000 \& 000 \&  |[fill=gray!15,inner sep=1pt]|001 \& 000 \& 000 \& 000 \& 000 \& 0111 \& 11111 \\[0.4ex]
      t_2^{(0)}: \&  000 \&  |[fill=gray!15,inner sep=1pt]|100 \& 000  \& |[fill=gray!15,inner sep=1pt]|100 \& 000 \& |[fill=gray!15,inner sep=1pt]|100 \&000 \& 000 \& 000 \& 0111 \& 11111 \\
      t_2^{(1)}: \&  000 \&  |[fill=gray!15,inner sep=1pt]|010 \& 000 \&  |[fill=gray!15,inner sep=1pt]|010 \& 000 \& |[fill=gray!15,inner sep=1pt]|010 \& 000 \& 000 \& 000 \& 0111 \& 11111 \\
      t_2^{(2)}: \&  000 \&  |[fill=gray!15,inner sep=1pt]|001 \& 000  \&  |[fill=gray!15,inner sep=1pt]|001 \& 000 \&  |[fill=gray!15,inner sep=1pt]|001 \& 000 \& 000 \& 000 \& 0111 \& 11111 \\[0.4ex]
      t_3^{(0)}: \&  000 \& 000 \&|[fill=gray!15,inner sep=1pt]|100 \&   |[fill=gray!15,inner sep=1pt]|100 \& 000 \& 000 \& |[fill=gray!15,inner sep=1pt]|100 \& 000 \& 000 \& 0111 \& 11111 \\
      t_3^{(1)}: \&  000 \& 000\& |[fill=gray!15,inner sep=1pt]|010 \&   |[fill=gray!15,inner sep=1pt]|010 \& 000 \& 000 \& |[fill=gray!15,inner sep=1pt]|010 \& 000 \& 000 \& 0111 \& 11111 \\
      t_3^{(2)}: \&  000 \& 000 \&|[fill=gray!15,inner sep=1pt]|001 \&   |[fill=gray!15,inner sep=1pt]|001 \& 000 \& 000 \& |[fill=gray!15,inner sep=1pt]|001 \& 000 \& 000 \& 0111 \& 11111 \\[0.4ex]
      t_4^{(0)}: \&  |[fill=gray!15,inner sep=1pt]|100 \&  000 \& |[fill=gray!15,inner sep=1pt]|100 \& 000 \& 000 \& 000 \& 000 \& |[fill=gray!15,inner sep=1pt]|100 \& 000 \& 0111 \& 11111 \\
      t_4^{(1)}: \&  |[fill=gray!15,inner sep=1pt]|010 \&  000 \&  |[fill=gray!15,inner sep=1pt]|010 \& 000 \& 000 \& 000 \& 000 \& |[fill=gray!15,inner sep=1pt]|010 \& 000 \& 0111 \& 11111 \\
      t_4^{(2)}: \&  |[fill=gray!15,inner sep=1pt]|001 \&  000 \&  |[fill=gray!15,inner sep=1pt]|001 \& 000 \& 000 \& 000 \& 000 \& |[fill=gray!15,inner sep=1pt]|001 \& 000 \& 0111 \& 11111 \\[0.4ex]
      t_5^{(0)}: \&  000 \& |[fill=gray!15,inner sep=1pt]|100  \& |[fill=gray!15,inner sep=1pt]|100 \& 000 \& 000 \& 000 \& 000 \& 000 \& |[fill=gray!15,inner sep=1pt]|100 \& 0111 \& 11111 \\
      t_5^{(1)}: \&  000 \& |[fill=gray!15,inner sep=1pt]|010  \&  |[fill=gray!15,inner sep=1pt]|010 \& 000 \& 000 \& 000 \& 000 \& 000 \&|[fill=gray!15,inner sep=1pt]|010 \& 0111 \& 11111 \\
      t_5^{(2)}: \&  000 \& |[fill=gray!15,inner sep=1pt]|001  \&  |[fill=gray!15,inner sep=1pt]|001 \& 000 \& 000 \& 000 \& 000 \& 000 \& |[fill=gray!15,inner sep=1pt]|001 \& 0111 \& 11111 \\[.6ex]
      \sol: \& 100 \& 010 \& 001 \& 100 \& 001 \& 001 \& 010 \& 010 \& 100 \& 0000 \& 00000\\ 
};

    \foreach \o / \n / \s / \t in {v/1/1/2,v/2/1/3,v/3/1/4,v/4/1/5,e/1/1/6,e/2/1/7,e/3/1/8,e/4/1/9,e/5/1/10} {
      \draw[decorate, decoration={brace,amplitude=3pt}]  ($(M-\s-\t.north west)+(0.12,0)$) -- node[above] {$\o_\n$} ($(M-\s-\t.north east)+(-0.12,0)$);
    }
    \foreach  \n / \s / \t in {1/1-1/3-1,2/4-1/7-1,3/8-1/22-1} {
      \draw[decorate, decoration={brace,amplitude=3pt,mirror}]  ($(M-\s.north)+(-.4,-0.1)$) -- node[midway, rotate=90,yshift=.5cm] {\footnotesize group $\n$} ($(M-\t.south)+(-.4,0.1)$);
    }
\end{tikzpicture}
 \end{figure}


    Finally, we define $k$ such that $k^p = \kvaluep$. This completes the construction, which can clearly be done in polynomial time. 

    \myproofparagraph{Correctness.}
  Before we show the correctness of our construction, we define a notion and make an observation.
  Let $s$ and $s'$ be two strings of equal length. 
  We say that \myemph{$s$ covers $s'$ exactly once} if there is exactly one integer $\ell \in \{1, 2, \cdots, |s|\}$ with $s[\ell]=s'[\ell] =1$.
  \begin{claim}[\appsymb]\label{claim:exact-cover-cost}
        Let $\sol$ and $s$ be two strings, both of length~$4\hat{n}$, such that 
    \begin{inparaenum}[(i)]
      \item $\sol$ has exactly $\hat{n}$ ones and each of them is in the first $3\hat{n}$~columns, and
      \item\label{claim-1:prop-s}in $s$, the first $3\hat{n}$~columns have exactly $3$ ones and the last $\hat{n}$~columns are $0\circ \one_{\hat{n}-1}$.
    \end{inparaenum}
    Then, if $\sol$ covers $s$ exactly once, 
    then $\hd{p}(\sol, s) + \hd{p}(\sol, \overline{s}) = \exactbound$;
    else $\hd{p}(\sol, s) + \hd{p}(\sol, \overline{s}) > \exactbound$.  
  \end{claim}
  \appendixproofwithstatement{claim:exact-cover-cost}{%
    Let $\sol$ and $s$ be two strings, both of length~$4\hat{n}$, such that 
    \begin{inparaenum}[(i)]
      \item $\sol$ has exactly $\hat{n}$ ones and each of them is in the first $3\hat{n}$~columns, and
      \item\label{claim-1:prop-s} in $s$, the first $3\hat{n}$~columns have exactly $3$ ones and the last $\hat{n}$~columns have the form $0\circ \one_{\hat{n}-1}$.
    \end{inparaenum}
    Then, if $\sol$ covers $s$ exactly once, 
    then $\hd{p}(\sol, s) + \hd{p}(\sol, \overline{s}) = \exactbound$;
    else $\hd{p}(\sol, s) + \hd{p}(\sol, \overline{s}) > \exactbound$.  
}{
  \begin{claimproof}[of \cref{claim:exact-cover-cost}]
    \ \\ 
    \noindent $\bullet$ Assume that $\sol$ covers~$s$ exactly once and let $\ell\in \{1,2,\ldots, 4\hat{n}\}$ be an integer with $\sol[\ell]=s[\ell]=1$.
    Since $\sol|^{4\hat{n}}_{3{\hat{n}}+1}=\zero_{\hat{n}}$, it follows that $\ell\in \{1,2,\ldots, 3\hat{n}\}$.
    By the property~\eqref{claim-1:prop-s} of $s$ in the claim,
    we have that $\hd{}(\sol, s)=\hd{}(\sol|^{3\hat{n}}_1, s|^{3\hat{n}}_1)+ \hd{}(\sol|^{4\hat{n}}_{3\hat{n}+1}, s|^{4\hat{n}}_{3\hat{n}+1})=(2+\hat{n}-1)+(\hat{n}-1)=2\hat{n}$.  
    Thus, $\hd{}(\sol, \overline{s})=4\hat{n}-\hd{}(\sol, {s})=2\hat{n}$.
    In summary, $\hd{p}(\sol, s) + \hd{p}(\sol, \overline{s}) = \exactbound$.

    \noindent $\bullet$ Assume that $\sol$ does not cover $s_i$ exactly once.
    If we can show that  $\hd{}(\sol, s)\neq 2\hat{n}$ holds, then
    since $\hd{}(s,\overline{s})=4\hat{n}$,
    by \cref{lem:complement}\eqref{lem-complement:3},
    we immediately obtain that
    $\hd{p}(\sol, s) + \hd{p}(\sol, \overline{s}) > \exactbound$.
    Thus, in the remainder of the proof, we only need to prove that
    $\hd{}(\sol, s)\neq 2\hat{n}$. 
    Since $s$ has exactly $3$ ones in the first $3\hat{n}$ columns, there are three cases to consider.
    \begin{compactitem}
      \item[Case 1:] For each $\ell\in \{1,2,\ldots, 3\hat{n}\}$, it holds that $\sol[\ell]+s[\ell]\le 1$.
      Consider the values in the first $3\hat{n}$ columns of $s$ and $\sol$:
      since $s$ contains exactly $3$ ones and $\sol$ contains exactly $\hat{n}$ ones, it follows that $\hd{}(\sol, s)=\hd{}(\sol|^{3\hat{n}}_1, s|^{3\hat{n}}_1)+ \hd{}(\sol|^{4\hat{n}}_{3\hat{n}+1}, s|^{4\hat{n}}_{3\hat{n}+1})=(\hat{n}+3)+(\hat{n}-1)=2\hat{n}+2$.
      \item[Case 2:] There are two distinct integers~$\ell, \ell'\in \{1,2,\ldots, 3\hat{n}\}$ with $\sol[\ell]=\sol[\ell']=s[\ell]=s[\ell']=1$ such that for each other integer~$j\in \{1,2,\ldots, 3\hat{n}\}\setminus \{\ell,\ell'\}$ it holds that $\sol[j]+s[j]\le 1$.
     By assumption, $\sol$ has exactly $\hat{n}$ ones in the first $3\hat{n}$ columns.
      Then, $\hd{}(\sol, s)=\hd{}(\sol|^{3\hat{n}}_1, s|^{3\hat{n}}_1)+ \hd{}(\sol|^{4\hat{n}}_{3\hat{n}+1}, s|^{4\hat{n}}_{3\hat{n}+1})=1+(\hat{n}-2)+(\hat{n}-1)=2\hat{n}-2$.
      \item[Case 3.] For each integer~$\ell \in \{1,2,\ldots, 3\hat{n}\}$ with $s[\ell]=1$ it holds that $\sol[\ell]=1$.
     By assumption, $\sol$ has exactly $\hat{n}$ ones in the first $3\hat{n}$ columns,
     and $s$ has exactly $3$ ones in the first $3\hat{n}$ columns.
     Thus, $\hd{}(\sol, s)=\hd{}(\sol|^{3\hat{n}}_1, s|^{3\hat{n}}_1)+ \hd{}(\sol|^{4\hat{n}}_{3\hat{n}+1}, s|^{4\hat{n}}_{3\hat{n}+1})=(\hat{n}-3)+(\hat{n}-1)=2\hat{n}-4$.
    \end{compactitem}%
    \hfill~(of \cref{claim:exact-cover-cost})~$\diamond$
  \end{claimproof}%
  }
  
  We show that $G$ has a proper $3$-coloring if and only if there is a string~$\sol$ 
  such that the sum of the $p$-distances from $\sol$ to all strings in $I'$ is at most $k^p= \kvaluep$.
  
  For the ``if'' direction, let $\sol$ be a string
  which has a sum of $p$-distances of at most $k^p$ to all strings in $I'$.
  Before we define a coloring for the vertices and show that it is proper
  we observe several properties of the solution string~$\sol$.
  
  By 
  \cref{lem:complement}\eqref{lem-complement:1}, the sum of $p$-distances 
  to all strings from group~2 and group~3
  is at least $2\cdot (2\hat{n})^p\cdot (n+3m)$
  since these groups consist of $n+3m$ pairs of strings, and the strings in each of these pairs have Hamming distance exactly $4\hat{n}$ to each other.
  By the definition of $k$,
  the sum of $p$-distances from $\sol$ to the first of group of strings is thus at most $\gvalue$. 
  Hence, by the contra-positive of \cref{lem:gadget-n}\eqref{lem-gadget:>},
  the solution string~$\sol$ has exactly $\hat{n}$ ones, 
  which all appear in the first $(2^b+1)\hat{n}$ columns, i.e.,
    $\hd{}(\sol,\zero_{(2^b+2)\hat{n}})=\hat{n} \text{ and } \hs(\sol,\zero_{(2^b+2)\hat{n}}) \subseteq [(2^b+1)\hat{n}].$ 
  By \cref{lem:gadget-n}\eqref{lem-gadget:=}, this implies that
  \begin{align}
    \sum_{s \in \text{group~1}}\hd{p}(\sol, s) = (2^a+2^{a-b})\cdot \hat{n}^{p}.\label{eq:bound-g1}
  \end{align}  
  Next, we claim that the ones in the solution~$\sol$ indeed all appear in the first $3\hat{n}$ columns,
  i.e., $\hs(\sol,\zero_{(2^b+2)\hat{n}}) \subseteq [3\hat{n}]$.
  Suppose, for the sake of contradiction, that solution~$\sol$
  contains $x$ ones which appear in columns ranging from $3\hat{n}+1$ to $(2^b+1)\hat{n}$ with $x>0$.
  Consider an arbitrary pair of strings~$s_i$ and $r_i$ from group~$2$ or an arbitrary pair of strings~$t^{(z)}_i$ and $w^{(z)}_i$ from group~$3$; for the sake of readability, represent them by $s$ and $s'$.
  By construction, strings~$s$ and $s'$ have Hamming distance exactly $4\hat{n}$ to each other, but have all zeros in the columns between $3\hat{n}+1$ and $(2^b+1)\hat{n}$.
  Since $x>0$, by the triangle inequality of Hamming distances, it follows that at least one string from the pair, $s$ or $s'$, has Hamming distance more than $2\hat{n}$ from $\sol$. %
  However, by 
  \cref{lem:complement}\eqref{lem-complement:3}, this means that the sum of $p$-distances from $\sol$ to $\{s,s'\}$ exceeds $2\cdot (2\hat{n})^p$.
  Since there are in total $n+3m$ such pairs in groups~2 and~3, the sum of $p$-distances from $\sol$ to these groups exceeds $2(n+3m)\cdot (2\hat{n})^p$,
  a contradiction to equation~\eqref{eq:bound-g1} and the defined bound~$k$.
  Thus, indeed, it holds that
  \begin{align}
    \hd{}(\sol,\zero_{(2^b+2)\hat{n}})=\hat{n} \text{ and } \hs(\sol,\zero_{(2^b+2)\hat{n}}) \subseteq [3\hat{n}].\label{eq:exact-sol}
  \end{align}
  This implies that, when determining the $p$-distance of $\sol$ to the strings from group 2 and group 3, we can ignore, the values in the columns ranging from~$3\hat{n}+1$ to $(2^b+1)\hat{n}$,
  in each string, which includes the solution~$\sol$, because $\sol$ also has only zeros in these columns. We will hence from now on treat these columns as if they do not exist.
In this way, we obtain strings of length~$4\hat{n}$.
Again, consider an arbitrary pair of strings~$s_i$ and $r_i$ from group~$2$ (resp.\ an arbitrary pair of strings~$t^{(z)}_i$ and $w^{(z)}_i$ from group~$3$), and represent them by $s$ and $s'$.
Since we ignore columns $3\hat{n}+1$ to $(2^b+1)\hat{n}$, string $s'$ is the complement of $s$.
  By construction, the Hamming distance between $s$ and $s'$ is exactly $4\hat{n}$.
  Using \cref{claim:exact-cover-cost} on $\sol,s,s'$ , the sum of $p$-distances from $\sol$
  to the pair~$\{s,s'\}$ is indeed equal to~$\exactbound$.
  By the same claim,  it follows that $\sol$ covers each string~$s_i$ (resp.\ $t^{(z)}_j$) from group~2 (resp.\ group~3) exactly once.

  Having this property, we are ready to color the vertices.
  Let $\col \colon V\to \{0,1,2\}$ be a mapping defined as follows.
  For each $v_i \in V$, set $\col(v_i) = z$ where $z\in \{0,1,2\}$ such that $\sol[3i-2+z]=1$.
  Note that, since $\sol$ covers $s_i$ exactly once and since $s_i$ has exactly three ones in the columns~$3i-2$, $3i-1$, and $3i$, 
  there is indeed such a $z$ with $\col(v_i)$.
  We claim that $\col$ is a proper coloring for $G$.
  Suppose, towards a contradiction, that there is an edge~$e_j=\{v_i, v_{i'}\}\in E$ such that 
  $v_i$ and $v_{i'}$ have the same color from $\col$, say $z\in \{0,1,2\}$.
  By the definition of $\col$, this means that $\sol[3i-2+z]=\sol[3i'-2+z]=1$.
  However, by the definition of the string~$t^{(z)}_j$ which corresponds to the edge~$e_j$,
  we also have that $t_j^{(z)}[3i-2+z]=t_j^{(z)}[3i'-2+z]=1$.
  This implies that $t^{(z)}_j$ is not covered by $\sol$ exactly once---a contradiction to our reasoning above that $\sol$ covers each string from the third group exactly once.

%
  \looseness=-1
    For the ``only if'' direction, let $\col\colon V\to \{0,1,2\}$ be a proper coloring for $G$.
  For an edge~$e\in E$ with two endpoints~$v_i, v_{i'}$, let $\col(e)=\{\col(v_i), \col(v_{i'})\}$.
  We claim that string~$\sol$, defined as follows, has the desired bound on the sum of the $p$-distances to all strings of $I'$.

  \qquad  \qquad$    \forall i  \in \{1,2,\cdots, n\}\colon
    \sol[3i-2,3i-1,3i] \coloneqq 
    \begin{cases}
        100, & \col(v_i) = 0,\\
        010, & \col(v_i) = 1,\\
        001, & \col(v_i) = 2.\\
      \end{cases}$

      $
        \forall j  \in \{n+1,n+2,\cdots, \hat{n}\}\colon
        \sol[3j-2,3j-1,3j] \coloneqq
      \begin{cases}
        100, & \col(e_j)=\{1,2\},\\
        010, & \col(e_j)=\{0,2\},\\
        001, & \col(e_j)=\{0,1\}.\\
      \end{cases}
      $

      \quad  \quad   \qquad   \qquad   \qquad \qquad\qquad\qquad\qquad ~ $\sol|^{(2^b+2)\hat{n}}_{3\hat{n}+1}\coloneqq \zero_{\hat{n}}.
   $
\looseness=-1

  First of all, since $\col$ is a proper coloring, $\sol$ is well defined in all $(2^b+2)\hat{n}$ columns. 
  Moreover, it has exactly $n$ ones in the first $3n$ columns and 
  exactly $m$ ones in the next $3m$~columns, 
  and all zeros in the remaining columns.
  Thus, by 
  \cref{lem:gadget-n}\eqref{lem-gadget:>}, the sum of the $p$-distances from $\sol$ to the first group of strings is $\gvalue$.

  Now, we focus on strings from group 2 and group 3.
  Since the solution~$\sol$ and each string in these groups have only zeros in the columns between $3\hat{n}+1$ and $(2^{b}+1)\hat{n}$,
  we can simply ignore the values in these columns and assume from now on that the strings have length~$4\hat{n}$.
  Moreover, for each $i\in [n]$, the pair $s_i$ and $r_i$ can be considered as complement to each other.
  Thus, for each string~$s_i$ from group~2, $\sol$ and $s_i$ fulfill the properties stated in \cref{claim:exact-cover-cost}.
  Moreover, by definition,~$\sol$ covers $s_i$ exactly once.
  Thus, by the same claim, we have that the sum of the $p$-distances from $\sol$ to all strings in group~2 is $n\cdot \exactbound$.

  Analogously, consider a string~$t^{(z)}_j$ from group~3, $j\in \{1,2,\ldots,m\}$ and $z\in \{0,1,2\}$.
  Recall that $t^{(z)}_j$ corresponds to the edge~$e_j$, 
  and let $v_i$ and $v_{i'}$ be the two endpoints of edge~$e_j$.
  We claim that $\sol$ covers $t^{(z)}_j$ exactly once.
  Observe that $t^{(z)}_j$ has exactly $3$ ones in the first $3\hat{n}$~columns, namely at columns~$3i-2+z$, $3i'-2+z$, and $3n+3j-2+z$.
   To prove that $s^*$  covers $t^{(z)}_j$ exactly once, 
   it suffices to show that $\sol$ has $1$ one in exactly one of these three columns.
  To show this, we consider the substrings~$t^{(z)}_j|_{3n+3j-2}^{3n+3j}$ and~$\sol|_{3n+3j-2}^{3n+3j}$. 
\mydescrip{Case 1:} $\sol|_{3n+3j-2}^{3n+3j}=t^{(z)}_j|_{3n+3j-2}^{3n+3j}$.
  By the definition of $\sol$, this implies that $\sol[3n+3j-2+z]=1$ and $\col(e_j)=\{0,1,2\}\setminus \{z\}$.
  We claim that $\sol[3i-2+z]=\sol[3i'-2+z]=0$.
  By the definition of $\sol$ regarding the columns that correspond to the endpoint~$v_i$ of edge~$e_j$, 
  we have that $\sol[3i-2+\col(v_i)]=1$ while $\sol[3i-2+z]=0$ (since $z\notin \col(e_j)=\{\col(v_i), \col(v_{i'})\}$).
  Analogously, by the definition of $\sol$ regarding the columns that correspond to the other endpoint~$v_{i'}$ of edge~$e_j$,
  we have that $\sol[3i'-2+\col(v_{i'})]=1$ while $\sol[3i'-2+z]=0$ (since $z\notin \col(e_j)=\{\col(v_i), \col(v_{i'})\}$).
    Thus, $3n + 3j - z$ is the only column in which both
    $\sol$ and $t^{(z)}_j$ have $1$ one, implying that $\sol$ covers
    $t^{(z)}_j$ exactly once.

    \mydescrip{Case 2:}  $\sol|_{3n+3j-2}^{3n+3j}\neq t^{(z)}_j|_{3n+3j-2}^{3n+3j}$.
    This means that $\sol[3n+3j-2+z]=0$ and that $z\in \col(e_j)$.
    To show that $\sol$ covers $t^{(z)}_j$ exactly once in this case, it suffices to show that
    either $\sol[3i-2+z]=1$ and $\sol[3i'-2+z]=0$, 
    or  $\sol[3i-2+z]=0$ and $\sol[3i'-2+z]=1$.

    \noindent $\bullet$ Assume that $\sol[3i-2+z]=1$.
    Then, by the definition of $\sol$ regarding the columns that correspond to the endpoint~$v_i$ of edge~$e_j$,
    this means that $\col(v_i)=z$.
    Since $\col$ is a proper coloring, it follows that $\col(v_{i'})\neq z$.
    Thus, again by the definition of $\sol$ regarding the columns that correspond to the other endpoint~$v_{i'}$ of edge~$e_j$, it follows that
    $\sol[3i'-2+z]=0$.

    \looseness=-1
    \noindent $\bullet$ Assume that $\sol[3i-2+z]=0$.
    Then, by the definition of $\sol$ regarding the columns that correspond to the endpoint~$v_i$ of edge~$e_j$,
    we have $\col(v_i)\neq z$.
    Since $z\in \col(e_j)$ and $\col$ is a proper coloring,
    the other endpoint~$v_{i'}$ of edge~$e_j$ must have color~$\col(v_{i'})=z$.
    Again, by the definition of $\sol$ regarding the columns that correspond  $v_{i'}$,
    it follows that 
    $\sol[3i'-2+z]=1$.


  \smallskip
  \noindent We have just shown that $\sol$ covers $t^{(z)}_j$ exactly once.
  Since $\sol$ and $t^{(z)}_j$ fulfill the property stated in \cref{claim:exact-cover-cost},
  it follows from the same claim that the sum of $p$-distances from $\sol$ to $t^{(z)}_j$ and to $w^{(z)}_j$ is $\exactbound$.
  There are $3m$ pairs in this group.
  So, the sum of the $p$-distances from $\sol$ to all strings of this group is $3m\cdot \exactbound$.
  
  In total, the sum of the $p$-distances from $\sol$ to all strings of $I'$ is $\gvalue+\exactbound\cdot (n+3m) = k^p$, as required.
\end{proof}
Our NP-hardness reduction implies the following running time lower bounds~\cite{CyFoKoLoMaPiPiSa2015}.
\newcommand{\corethlowerbound}{%
  For each fixed rational number~$p>1$, unless the \ETH\ fails, no $2^{o(\hat{n}+\hat{m})}\cdot |I'|^{O(1)}$-time 
  or $2^{o(k^{\nicefrac{p}{(p+1)}})}\cdot |I'|^{O(1)}$-time 
  algorithm exists that decides every given instance~$I'$ of \pHDC{} where $\hat{n}$ is the length of the input strings,
$\hat{m}$ is the number of input strings,
and $k$ is the $p$-norm bound.%
}
\begin{corollary}[\appsymb]\label[corollary]{cor:ETH-lowerbound}
  \corethlowerbound
\end{corollary}
\appendixproofwithstatement{cor:ETH-lowerbound}{\corethlowerbound}{
\begin{proof}
  Let $a$ and $b$ be two fixed coprime integers such that $p=a/b$.
  To show our statement, note that we have constructed $2^{a-b}+1+2(n+3m)$ strings for our \pHDClong{} problem in the proof for \cref{thm:mphd-np-h}, each of which has length $(2^b+2)\cdot (n+m)$, where  $n$ and $m$ are the number of vertices and the number edges in the instance of \TCol{}.
  The $p$-norm bound~$k$ was set to $\kvalue$ which is upper-bounded by $\big((1+6\cdot 2^{p})\cdot (n+m)\big)^{\frac{p+1}{p}}$ since $a$ and $b$ are fixed integers.
  Thus, a $2^{o(\hat{n}+\hat{m})}\cdot |I'|^{O(1)}$-time or a $2^{o(k^{\nicefrac{p}{(p+1)}})}$-time algorithm for \pHDClong{} implies a $2^{o(n+m)}\cdot {(n\cdot m)}^{O(1)}$-time algorithm for \TCol{}, which is unlikely unless the \ETH\ fails~\cite[Theorem 14.6]{CyFoKoLoMaPiPiSa2015}.
\end{proof}
}
\noindent Using a slight modification of the construction, we can show
that our results are not idiosyncratic to instances which contain some
strings multiple times. (Recall that the gadget from
\cref{lem:gadget-n} in the construction contains $2^{a-b}$ copies of
the all-zero string.) The basic idea is to append an identity matrix
to the strings we need to distinguish, and then to show using a
slightly more involved analysis that the gadgets still work in the
same way.
\newcommand{\propnphdistinct}{
  \Cref{thm:mphd-np-h,cor:ETH-lowerbound} hold even if all input strings are distinct.
}
\begin{proposition}[\appsymb]\label{prop:nph-distinct}
\propnphdistinct
\end{proposition}
\appendixproofwithstatement{prop:nph-distinct}{\propnphdistinct}{
\begin{proof}
  Again, let $a$ and $b$ be two fixed coprime integers such that $p=a/b$.
  To show the statement, we modify the instance that we constructed in the proof of \cref{thm:mphd-np-h} by appending to each string $2^{a-b}+2^b$ columns.
  First, observe that it suffices to distinguish all $2^{a-b}$ all-zero strings in group~1 from one another: All other strings are distinct. 
  We need to preserve, however, the property of the gadget in group~1.
  To do that, intuitively, we attach an identity matrix to the strings in
  group~1, and fill up the remaining strings (in group~2 and group 3) with zeros.
  
  More formally, let $g_0, \ldots, g_{2^{a-b}}$ be the strings
  in group~1, where $g_0$ is the single string with exactly $(2^b+1)\hat{n}$~ones.
  Append to string~$g_0$ the string~$\zero_{2^{a-b}} \circ \one_{2^b}$.
  For each string~$g_i$, $i \in [2^{a-b}]$, append to it the string $\zero_{i-1} \circ 1 \circ \zero_{2^{a-b} + 2^b - i}$.
  Append an all-zero string~$\zero_{2^{a-b}+2^b}$ to each string from group~2 and group~3,
  i.e.,
  to each string~$s_i, r_i$, $i \in [n]$ and each string~$t^{(0)}_{i'}, t^{(1)}_{i'}, t^{(2)}_{i'}, w^{(0)}_{i'},  w^{(1)}_{i'}, w^{(2)}_{i'}$, $i' \in [m]$.
  See \cref{fig:distinct-strings} for an illustration.
  \begin{figure}
    \caption{Illustration of an instance constructed in the proof of \cref{prop:nph-distinct}, where each string is distinct. This instance is obtained by appending to the instance constructed in the proof of \cref{thm:mphd-np-h} some appropriately designed columns (see the last $2^{a-b}+2^b$ columns).}\label{fig:distinct-strings}
    \centering
    \begin{tikzpicture}
      \def \msize {6ex}
      \node[text width = \msize, align=right] (g0) {$g_0:$};
      
      \node[draw, rectangle, right = 4ex of g0, anchor=center] (g01) {$1 1 \cdots 1$};
      \node[draw, rectangle, right = 1ex of g01] (g02) {$1 1 \cdots \cdots 1$};
      \node[draw, rectangle, right = 1ex of g02] (g03) {$0 0 \cdots 0$};
      \node[draw, rectangle, right = 1ex of g03] (g04) {$0 0 \cdots 0$};
      \node[draw, rectangle, right = 1ex of g04] (g05) {$1 1 \cdots 1$};

      \draw[decorate, decoration={brace,amplitude=6pt}] ($(g01.north west)+(0,0.1)$) --  ($(g01.north east)+(0,0.1)$)  node[pos=0.5, above, yshift=2.5ex, anchor=center]  {\small $3\hat{n}$};

      \draw[decorate, decoration={brace,amplitude=6pt}] ($(g02.north west)+(0,0.1)$) --  ($(g02.north east)+(0,0.1)$)  node[pos=0.5, above, yshift=2.5ex, anchor=center] {\small $(2^b-2)\hat{n}$};

      \draw[decorate, decoration={brace,amplitude=6pt}] ($(g03.north west)+(0,0.1)$) --  ($(g03.north east)+(0,0.1)$)  node[pos=0.5, above, yshift=2.5ex, anchor=center] {\small $\hat{n}$};

      \draw[decorate, decoration={brace,amplitude=6pt}] ($(g04.north west)+(0,0.1)$) --  ($(g04.north east)+(0,0.1)$)  node[pos=0.5, above, yshift=2.5ex, anchor=center] {\small $2^{a-b}$};
      
      \draw[decorate, decoration={brace,amplitude=6pt}] ($(g05.north west)+(0,0.1)$) --  ($(g05.north east)+(0,0.1)$)  node[pos=0.5, above, yshift=2.5ex, anchor=center]  {\small $2^{b}$};

      \node[below = 1ex of g0, text width=\msize, align=right] (g1) {$~~g_1:$};
      \node[draw, rectangle, right = 4ex of g1, anchor=center] (g11) {$0 0 \cdots 0$};
      \node[draw, rectangle, right = 1ex of g11] (g12) {$0 0 \cdots \cdots 0$};
      \node[draw, rectangle, right = 1ex of g12] (g13) {$0 0 \cdots 0$};
      \node[draw, rectangle, right = 1ex of g13] (g14) {$1 0 000 0$};
      \node[draw, rectangle, right = 1ex of g14] (g15) {$0 0 \cdots 0$};

      \node[below = 0ex of g1, text width=\msize, align=center] (vdots) {$\vdots$};
      \node[below = 0ex of vdots, text width=\msize, align=right] (gn) {$g_{2^{a-b}}:$};
      \node[below = 1ex of gn, text width=\msize, align=right] (s1) {$s_1:$};
      \node[below = 0ex of s1, text width=\msize, align=center] (vds) {$\vdots$};
      \node[below = 0ex of vds, text width=\msize, align=right] (wm2) {$w_m^{(2)}:$};

      \def \ys {12}
      \def \yss {63}
      {
        \gettikzxy{(g01.west)}{\x}{\y}
        \gettikzxy{(g03.east)}{\xx}{\yy}

        \draw[fill=white] (\x, \y-\ys) rectangle  node {\Huge $0$} (\xx, \yy-\yss);
      }
      
      {
        \gettikzxy{(g04.west)}{\x}{\y}
        \gettikzxy{(g04.east)}{\xx}{\yy}

        \draw[fill=white] (\x, \y-\ys) node[below, xshift=1ex] (ul) {$1$} rectangle (\xx, \yy-\yss) node[above, xshift=-1ex] (lr) {$1$};

        \draw (ul) edge (lr);
        \draw (\xx, \y-\ys) node[below, xshift=-2ex] (ur) {\huge $0$} rectangle (\x, \yy-\yss) node[above, xshift=2ex] (ll) {\huge $0$};
      }

      {
        \gettikzxy{(g05.west)}{\x}{\y}
        \gettikzxy{(g05.east)}{\xx}{\yy}

        \draw[fill=white] (\x, \y-\ys) rectangle node  {\Huge $0$}  (\xx, \yy-\yss);
      }

        \def \ys {69}
        \def \yss {124}
        \gettikzxy{(s1)}{\sx}{\sy}

        {
        \gettikzxy{(g01.west)}{\x}{\y}
        \gettikzxy{(g01.east)}{\xx}{\yy}

       \draw[fill=white] (\x, \y-\ys)  rectangle (\xx, \yy-\yss);

       \node at (\x*0.5+\xx*0.5, \sy) (u1) {$u_1$};
       \node[below=0ex of u1] (vdds) {$\vdots$};
       \node[below=0ex of vdds] (em) {$e_m^{(2)}$};
     }
     
     {
        \gettikzxy{(g02.west)}{\x}{\y}
        \gettikzxy{(g02.east)}{\xx}{\yy}

        \draw[fill=white] (\x, \y-\ys) rectangle node[]  {\Huge $0$}  (\xx, \yy-\yss);
     }

     {
        \gettikzxy{(g03.west)}{\x}{\y}
        \gettikzxy{(g03.east)}{\xx}{\yy}

        \draw[fill=white] (\x, \y-\ys) rectangle  (\xx, \yy-\yss);
       \node at (\x*0.5+\xx*0.5, \sy) (u1) {$01\cdots 1$};
       \node[below=0ex of u1] (vddds) {$\vdots$};
       \node[below=1ex of vddds] (em) {$01\cdots 1$};

     }

        {
        \gettikzxy{(g04.west)}{\x}{\y}
        \gettikzxy{(g05.east)}{\xx}{\yy}

        \draw[fill=white] (\x, \y-\ys) rectangle node[]  {\Huge $0$}  (\xx, \yy-\yss);

      }

    \end{tikzpicture}
  \end{figure}
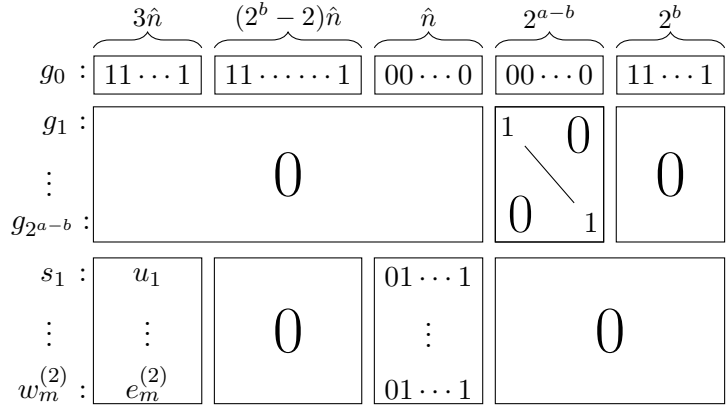
  Finally, we set $k$ to the positive real so that $k^p = \guniqvalue+\distance$; recall that $\hat{n}=n+m$. 
  Note that $k=O((n+m)^{\frac{p+1}{p}})$ still holds as $a$ and $b$ are fixed integers.
  
  For ease of notation we use the overloaded symbols $g_0,g_1,\ldots,g_{2^{a-b}}, s_1,\ldots,s_n$, $r_1,\ldots,r_{n}$, $t^{(z)}_1,\ldots$, $t^{(z)}_{m}, w^{(z)}_{1}, \ldots, w^{(z)}_{m}, z\in \{0,1,2\}$, to refer to the modified strings.
  
  To show that the construction remains correct, we first claim that an arbitrary solution has sum of $p$-distance at least $\guniqvalue$ to the strings of the first group.

    \begin{claim}[\appsymb]\label{claim:mind-cost-group1}
      Let $\sol$ be an arbitrary solution string,
      then the sum of $p$-distances from $\sol$ to all strings of group~$1$ is at least $\guniqvalue$.
    \end{claim}

    \begin{claimproof}[of \cref{claim:mind-cost-group1}]

      Let $x$ denote the number of ones of solution~$\sol$ in the columns of $\{1,\ldots, (2^{b}+1)\cdot \hat{n}, (2^{b}+2)\hat{n}+2^{a-b}+1,  \ldots, (2^{b}+2)\hat{n}+2^{a-b}+2^b\}$ with $0\le x \le (2^{b}+1)\hat{n}+2^{b}$.
      To show the statement, we distinguish between two cases,
      depending on whether $\sol$ contains a one in the column range~$[(2^b+2)\hat{n}+1, (2^b+2)\hat{n}+2^{a-b}]$.
      
      \mydescrip{Case 1:}~$\hs(\sol|_{(2^b+2)\hat{n}+1}^{(2^b+2)\hat{n}+2^{a-b}}, \zero_{2^{a-b}})\neq \emptyset$, that is, $\sol$ contains a one in the column range~$[(2^b+2)\hat{n}+1, (2^b+2)\hat{n}+2^{a-b}]$.
        In this case, it holds that $\hd{}(\sol, g_0) \ge (2^b+1)\hat{n}+2^b-x+1$ and
        for each $i\in [2^{a-b}]$ it holds that $\hd{}(\sol, g_i) \ge x$.
        So, the sum of the $p$-distances between $\sol$ and the strings of the first group is at least:
      \begin{align*}
        \sum_{i=0}^{2^{a-b}}\hd{p}(g_i,\sol) \ge 2^{a-b}\cdot x^p
        + ((2^b+1)\hat{n}+2^b-x+1)^p =: f(x).
      \end{align*}
      To derive a lower bound on the above cost, we use a proof similar to one given for \cref{lem:gadget-n},
      but utilizing the first and the second derivatives of $f(x)$:
      \begin{align*}
           \frac{df}{dx} & =  p\cdot 2^{a-b}\cdot x^{p - 1} -
                           p\cdot ((2^b+1)\hat{n}+2^b-x+1)^{p-1}   \\
                         & = p \cdot \Big( (2^bx)^{p-1}  - \big((2^b+1)\hat{n} + 2^b +1 - x\big)^{p-1}\Big).                           
      \end{align*}
      Now observe that the first derivative shown above is zero when
      $(2^bx)^{p-1}  - \big((2^b+1)\hat{n} + 2^b +1 - x\big)^{p-1}= 0$ because $p>1$.
      Solving the equation, the first derivative is zero when $x = \hat{n}+1$.
      Furthermore, the second
      derivative $\frac{d^2f}{dx^2}$ is positive at $x = \hat{n} + 1$ since $p>1$, 
      meaning that $f(x)$ has a local minimum at this point.
      The minimum value is thus $f(\hat{n}+1)=2^{a-b}(\hat{n}+1)^p+(2^b(\hat{n}+1))^{p}=(2^{a-b}+2^a)\cdot (\hat{n}+1)^p$.

   \mydescrip{Case 2:}
   Analogously, we consider the case when~$\hs(\sol|_{(2^b+2)\hat{n}+1}^{(2^b+2)\hat{n}+2^{a-b}}, \zero_{2^{a-b}})=\emptyset$. 
   In this case, it holds that $\hd{}(\sol, g_0) \ge (2^b+1)\hat{n}+2^b-x$ and
   for each $i\in [2^{a-b}]$ it holds that $\hd{}(\sol, g_i) \ge x+1$.
   Thus, the sum of the $p$-distances between $\sol$ and the strings of the first group is at least:
      \begin{align*}
        \sum_{i=0}^{2^{a-b}}\hd{p}(g_i,\sol) \ge 2^{a-b}\cdot (x+1)^p
        + ((2^b+1)\hat{n}+2^b-x)^p =: g(x).
      \end{align*}
      To derive a lower bound on the above cost, we use a proof similar to one given for \cref{lem:gadget-n},
      but utilizing the first and the second derivatives of $g(x)$:
            \begin{align*}
           \frac{dg}{dx} & =  p\cdot 2^{a-b}\cdot (x+1)^{p - 1} -
                           p\cdot ((2^b+1)\hat{n}+2^b-x)^{p-1}   \\
                         & = p \cdot \Big( \big(2^b\cdot (x+1)\big)^{p-1}  - \big((2^b+1)\hat{n}+2^b-x\big)^{p-1} \Big).
            \end{align*}
      Now observe that the first derivative shown above is zero when
      $\big(2^b\cdot (x+1)\big)^{p-1}  - \big((2^b+1)\hat{n} + 2^b - x\big)^{p-1}= 0$ because $p>1$.
   Solving the equation, the first derivative is zero when $x = \hat{n}$.
   Furthermore, the second
   derivative $\frac{d^2g}{dx^2}$ is positive at $x = \hat{n}$ (note that $p>1$),
   meaning that $g(x)$ has a local minimum at this point.
   The minimum value is thus $g(\hat{n})=2^{a-b}(n+1)^p+(2^b(\hat{n}+1))^{p}=(2^{a-b}+2^{a})\cdot (\hat{n}+1)^{p}$.

   \medskip
   
   \noindent  Summarizing, the sum the $p$-distances from $\sol$ to all strings from group~1 is at least $\guniqvalue$. \hfill~{(of \cref{claim:mind-cost-group1})~$\diamond$}
 \end{claimproof}

    Now, we prove that any solution string where the last $2^{a-b}+2^b$ columns have at least $1$ one will exceed our cost~$k$.
  \begin{claim}\label{claim:bound-exceed}
    Let $\sol$ be a solution with~$\sol|_{(2^{b}+2)\hat{n}+1}^{(2^b+2)\hat{n}+2^{a-b}+2^b}\neq \zero_{2^{a-b}+2^b}$, then
    the sum of $p$-distances from $\sol$ to the modified strings is larger than $k^p$.
  \end{claim}

  \begin{claimproof}[of \cref{claim:bound-exceed}]
    We derive the $p$-distances to group~$1$, and to groups~$2$ and $3$, separately.
    From \cref{claim:mind-cost-group1}, stating that the sum of $p$-distances to group~1 is at least $\guniqvalue$.

   Now, let us consider an arbitrary pair~$s_i$ and $r_i$ (resp.\ $t^{(z)}_j$ and $w^{(z)}_j$) of strings, representing them by $s$ and $s'$.
    By construction, it holds that $\hd{}(s,s')=4\hat{n}$ and $\hs(s,s')\subseteq \{1,\ldots,3\hat{n},(2^{b}+1)\hat{n}+1,\ldots, (2^{b}+2)\hat{n}\}$.
    However, by assumption that $\sol|_{(2^{b}+2)\hat{n}+1}^{(2^b+2)\hat{n}+2^{a-b}+2^b}\neq \zero_{2^{a-b}+2^b}$, at least one of the strings from $\{s,s'\}$ has Hamming distance more than $2\hat{n}$ to $\sol$.
    By 
    \cref{lem:complement}\eqref{lem-complement:3}, it follows that
    $\hd{p}(\sol,s)+\hd{p}(\sol,s')> 2\cdot (2\hat{n})^p$.
    Since we have $n+3m$ such pairs from group~2 and group~3,
    it follows that the sum of $p$-distances to group~2 and group~3 is more than $\distance$.

    In total, the sum of $p$-distances from such a string~$\sol$ to the modified strings exceeds $\guniqvalue+\distance$.%
    \hfill~{(of \cref{claim:bound-exceed})~$\diamond$}  \end{claimproof}

  By \cref{claim:bound-exceed}, we may assume that the last $2^{a-b} + 2^b$ columns in a
  solution (with cost $k^p$) contain only zeros, it now follows that the Hamming
  distance of a solution to each string in the constructed instance in
  the proof of \cref{thm:mphd-np-h} remains the same after our
  modifications---except for those distances that relate to the $g_i$.
 It hence remains to show that an analog of \cref{lem:gadget-n}
  remains valid in which the gadget's strings are appended with an
   identity matrix as above and $\sol$ contains only zeros in the last
   $2^{a-b} + 2^b$ columns.
    \begin{claim}\label{claim:last-columns-zeros}
      Let $\sol\in \{0,1\}^{(2^{b}+2)\hat{n}+2^{a-b}+2^b}$ be a solution with $\sol|_{(2^{b}+2)\hat{n}+1}^{(2^b+2)\hat{n}+2^{a-b}+2^b}=\zero_{2^{a-b}+2^b}$.
     Then the following holds.
       \begin{enumerate}[(1)]
         \item If $\hs(\sol,\zero_{(2^b+2)\hat{n}+2^{a-b}+2^b}) \subseteq [(2^b+1)\cdot \hat{n}]$
         and $\hd{}(\sol,\zero_{(2^b+2)\hat{n}+2^{a-b}+2^b}) = \hat{n}$,
         then $\sum_{i=0}^{2^{a-b}}\hd{p}(\sol, g_i)= \guniqvalue$.
    \item Otherwise, $\sum_{i=0}^{2^{a-b}}\hd{p}(\sol, g_i) > \guniqvalue$.
  \end{enumerate}
\end{claim}

\begin{claimproof}[of \cref{claim:last-columns-zeros}]
  The first statement follows by a straight-forward calculation.
  
  The proof for the second statement is analogous to the one given for \cref{lem:gadget-n}.
  Again, let $y$ denote the number of ones in~$\sol$ in the first $(2^b+1)\cdot \hat{n}$ columns.  
 Then, $\sum_{i=0}^{2^{a-b}}\hd{p}(\sol, g_i) \ge ((2^{b}+1)\hat{n}-y+2^b)^{p} + 2^{a-b}\cdot (y+1)^p$; note that by assumption, $\sol$ has zeros in the last $2^{a-b}+2^{b}$~columns.
 We define a function
 \[f \colon [0,(2^b+1)\cdot \hat{n}] \to \mathds{Z} \text{ with } f(y)\coloneqq  ((2^{b}+1)\hat{n}-y+2^b)^{p} + 2^{a-b}\cdot (y+1)^p,\]
  and show that this function attains its sole integer minimum over $[0, (2^b+1)\cdot \hat{n}]$ at $y=\hat{n}$.
  Note that this function is a lower bound on the sum of $p$-distances of $\sol$ to the first group of strings.
  First of all, the first derivative of $f$ with respect to~$y$ is 
  \begin{align*}
    \frac{d f}{d y} & = -p\cdot \big((2^{b}+1)\hat{n}-y+2^b\big)^{p-1} + p\cdot2^{a-b}\cdot (y+1)^{p-1}\\
                    & = p \cdot \Big( 2^{a-b}\cdot (y+1)^{p-1} -
                      \big((2^b+1)\cdot \hat{n} + 2^b -y\big)^{p-1}
                      \Big).
  \end{align*}
  Now, observe that the first derivative is zero only when the second component shown above is zero:
$2^{a-b}\cdot (y+1)^{p-1} -
                      \big((2^b+1)\cdot \hat{n} + 2^b -y\big)^{p-1} = 0$ because $p>1$.
Solving the equation, we obtain that the first derivative is zero when $y = \hat{n}$.
The second derivative of $f$ respect to $y$ is
 \[\frac{d^2f}{dy^2}=
   p \cdot (p-1)\cdot \Big(
   2^{a-b}\cdot (y+1)^{p-2} + \big( (2^b+1)\cdot \hat{n} - y + 2^b \big)^{p-2}
   \Big),\]
which is positive at $y=\hat{n}$
(recall that $p > 1$).
Hence, indeed, the sole
minimum of $f(y)$ over $[0, (2^b + 1)\cdot \hat{n}]$ is attained at $y = \hat{n}$ with
$f(\hat{n})= \big(2^b\cdot (\hat{n}+1)\big)^p + 2^{a-b}\cdot (\hat{n}+1)^p = (2^a+a^{a-b})\cdot (\hat{n}+1)^p$.

To summarize,
if $\hs(\sol,\zero_{(2^b+2)\hat{n}+2^{a-b}+2^b}) \subseteq [(2^b+1)\cdot \hat{n}]$
but $\hd{}(\sol,\zero_{(2^b+2)\hat{n}+2^{a-b}+2^b}) \neq \hat{n}$,
then $\sum_{i=0}^{2^{a-b}}\hd{p}(\sol,g_i) = f(y) > f(\hat{n})=(2^a+2^{a-b})\cdot (\hat{n}+1)^p$.
If $\hs(\sol,\zero_{(2^b+2)\hat{n}})\nsubseteq [(2^b+1)\hat{n}]$,
then $\sol$ has at least $1$ one in the column range~$[(2^b+1)\cdot \hat{n}+1, (2^b+2)\cdot \hat{n}]$.
Since each string from the first group has zeros in all these columns,
it follows that $ \sum_{i=0}^{2^{a-b}}\hd{p}(\sol,g_i)  \ge
 ((2^{b}+1)\hat{n}-y+2^b+1)^{p} + 2^{a-b}\cdot (y+2)^p> f(y) \ge f(\hat{n})=(2^a+2^{a-b})\cdot (\hat{n}+1)^p$.%
~\hfill~{(of \cref{claim:last-columns-zeros})~$\diamond$}\end{claimproof}
 \noindent  By the above claim, the correctness of our modified construction follows immediately. 

  As for the lower bound, no additional string is added to the new construction, and the length of the modified strings is increased by $2^{a-b}+2^b$, which is a constant.
  Moreover, as already observed, $k\in O((n+m)^{\frac{p+1}{p}})$.

  Altogether, we obtain the same ETH-based lower bounds, even if all input strings are distinct.
\end{proof}
}

  
  Let \textsc{$p$-Norm Approval Committee} be the variant of \pHDC\ in
  which we additionally get $t \in \mathds{N}$ as an input and require the number of ones in the solution
  string~$\sol$ to be exactly~$t$ \cite{sivarajan_generalization_2018}.
  Note that in the proof of \cref{thm:mphd-np-h} we have first shown
  that each solution string contains exactly~$\hat{n}$ ones. Thus, the
  reduction works in the same way for \textsc{$p$-Norm Approval
    Committee} when we specify $t = \hat{n}$ in the constructed
  instance. We hence obtain the following.

  \begin{corollary}\label{cor:committee-nph}
    For each fixed rational~$p > 1$, \textsc{$p$-Norm Approval
      Committee} is NP-hard and admits no algorithm
    running in $2^{o(\hat{n}+\hat{m})}\cdot |I'|^{O(1)}$-time or in
    $2^{o(k^{\nicefrac{p}{(p+1)}})}\cdot |I'|^{O(1)}$-time unless the
    \ETH\ fails, where $\hat{n}$ is the number of candidates,
    $\hat{m}$ is the number of voters, and $k$ is the $p$-norm bound.
  \end{corollary}
\section{Algorithmic results}\label{sec:algo}
We now turn to our positive results. In \cref{sec:subexp} we provide
an efficient algorithm when the objective value~$k$ is small.
In \cref{sec:FPT-m}, we derive an integer convex
programming formulation to obtain an efficient algorithm for instances where
the number~$m$ of input strings is small. Finally, we give a simple
2-approximation in \cref{sec:2approx}.

\subsection{A subexponential-time algorithm}\label{sec:subexp}
\appendixsection{sec:subexp} In this section, we present an algorithm
with running
time~$2^{k^{\nicefrac{p}{(p+1)}+\epsilon}}\cdot |I|^{O(1)}$ for any
$\epsilon > 0$ and input instance~$I$ with distance bound~$k$. By the
lower bound result given in \cref{cor:ETH-lowerbound}, we know that
under ETH, the running time of the obtained algorithm is tight.

The algorithm is built on two subcases, distinguishing on a relation between the number~$m$ of input strings and the distance bound~$k$.
In each subcase we use a distinct algorithm that runs in subexponential time when restricted to that subcase.
To start with,
a dynamic programming algorithm which keeps track of the achievable vector of Hamming distances to each input string after columns $1$ to $j \leq n$ has running time~$O(n\cdot k^m)$.




\newcommand{\lemdpm}{%
  \pHDC{} can be solved in $O(n\cdot k^m)$ time and space, where $m$ and $n$ are the number and the length of the input strings, respectively, and $k$ is the p-norm distance bound.
}
\begin{lemma}[\appsymb]\label[lemma]{lem:dp-m}
\lemdpm
\end{lemma}
\appendixproofwithstatement{lem:dp-m}{\lemdpm}{
\begin{proof}
  Let $I=(S,k)$ be an instance of \pHDC{} with $S=(s_1,\dots,s_m)$ being the input strings of length $n$ and $k$ being the p-norm distance bound.
  First of all, it is obvious that if $I$ is a yes-instance and $\sol$ is a solution for $I$, meaning that $\objectivel{S}{\sol} \le k^p$,
  then the Hamming distance between $\sol$ and each input string~$s_i\in S$ must not exceed $k$. To ease notation and slightly improve the running time, we reduce to the case where this distance does not exceed~$k - 1$. Indeed, if there is an input string $s_i$ such that the $p$-distance between $\sol$ and $s_i$ is exactly~$k$, then there is another input string~$s_j$ such that the $p$-distance between~$\sol$ and~$s_j$ is zero. We can check whether there exists a solution which is equal to some input string in $O(nm)$ time. Thus, we reduce to the case where the $p$-distance between $\sol$ and each input string~$s_i$ does not exceed~$k - 1$.

  Based on the above observation, we can design a dynamic program that keeps track, for each $m$-tuple of Hamming distances, whether there is a partial solution that ``fulfills'' these Hamming distances.
  More precisely, our dynamic-programming table~$T$ stores 
  for each $m$-tuple~$(d_1,\ldots,d_m) \in \{0, 1, \ldots, k - 1\}^m$ and each column index~$j$, 
  whether there is a partial solution of length~$j$ 
  that has Hamming distance~$d_i$ to each input string~$s_i$
  when restricted to only the first $j$ columns.

  For each tuple~$D=(d_1,\ldots, d_m)\in \{0, 1, \ldots, k - 1\}^{m}$,
  we set $T(D,1)=\textsf{true}$ if $D=(s_i[1])_{1\le i \le m}$ or  $D=(1-s_i[1])_{1\le i \le m}$ and $T(D,1)=\textsf{false}$ otherwise.
  Then, for each column index~$j \ge 2$ in increasing order, 
  we set $T(D,j)= T(D_1,j-1) \vee T(D_2,j-1)$ where $D_1,D_2\in \{0, 1, \ldots, k - 1\}^m$
  such that $D_1 = (d_i - s_i[j])_{1\le i \le m}$ and $D_2 = (d_i - ( 1 - s_i[j]))_{1\le i \le m}$ (if $D_1$ or $D_2$ does not exist, we replace the corresponding table entry $T(D_r, j - 1)$, $r \in \{1, 2\}$, with $\textsf{false}$ in the formula for $T(D, j)$).
  Intuitively, 
  $D_1$ (resp.\ $D_2$) corresponds to setting the $i^{\text{th}}$ column of a solution to zero (resp.\ one).
  Since setting the $i^{\text{th}}$ column of a solution to zero (resp.\ one) will increase the Hamming distance of an input string that has a one (resp.\ a zero) in this column,
  we should update the Hamming distances accordingly.
  Finally, our input instance is a yes-instance if and only if there is a tuple~$(d_1,\ldots, d_m)\in \{0, 1, \ldots, k - 1\}^m$ with $\sum_{1\le i\le m}d^p_i\le k^p$ such that $T(d_1,\ldots, d_m,n)=\text{true}$.
 The running time and space are $O(k^m\cdot n)$ since the dynamic table has $k^m\cdot n$ entries and each entry can be computed in constant time.
\end{proof}
}

\looseness=-1
The dynamic program given in \cref{lem:dp-m} is efficient if there is a small number~$m$ of input strings only. 
In particular, if $m$ satisfies $m\le \frac{k^{\nicefrac{p}{(p+1)}}}{\log{k}}$, then we immediately obtain an $O(n\cdot 2^{k^{\nicefrac{p}{(p+1)}}})$-time algorithm.
Otherwise, we can use \cref{lem:XP-n}.
The algorithm behind \cref{lem:XP-n} is based on a different but related idea as the fixed-parameter algorithm for \CloStr\ given by \citet{GraNieRoss2003}: We use data reduction to shrink the length of the strings by~$k^p$, observe that one of the input strings must be close to a solution with bound~$k$ if it exists, and then find the solution by a search tree.
\begin{lemma}\label[lemma]{lem:XP-n}
  \pHDC{} can be solved in $O(nm^2 \cdot k^{\frac{p\cdot k}{\sqrt[p]{m}}})$~time, where $m$ and $n$ are the number and the length of the input strings, respectively, and $k$ is the p-norm distance bound.
\end{lemma}

\begin{proof}
  \looseness=-1 Let $I=(S,k)$ be an instance of \pHDC{} with $S=(s_1,\dots,s_m)$ being the input strings of length $n$ and $k$ being the p-norm distance bound.
  To show the statement, we first observe that if a column is an all-zero (resp.\ an all-one) column,
  then we can simply assume that an optimal solution will also have zero (resp.\ one) in this column as our objective function is convex. 
  By preprocessing all columns that are either an all-zero or an all-one vector, 
  we obtain an equivalent instance, where each column has at least a zero and at least a one.
  Thus, for each column, no matter which value a solution has at this column, it will always induce Hamming distance of at least one to some input string. 
  Consequently, if there are more than $k^p$ columns remaining,
  then we can simply answer ``no'' as any string will have cost more than $k$ to the input.
  Otherwise, there remain at most $k^p$ columns.

  If $I$ is a yes-instance, meaning that there is a solution~$\sol$
  for $I$ with $\|(\sol,S)\|_p \le k$, then there is an
  input string~$s^{**}\in S$ whose Hamming distance satisfies
  $\hd{}(s^{**},\sol)\le \sqrt[p]{\frac{k^p}{m}} =
  \frac{k}{\sqrt[p]{m}}$. Thus, we iterate over all input strings
  in~$S$, assuming in each iteration that the current string is the aforementioned~$s^{**}$. For each string $s_i$ that we assume to be the aforementioned~$s^{**}$, we go over all strings
  $\hat{s}$ that differ from $s_i$ by $k'$ columns with
  $k'\le \frac{k}{\sqrt[p]{m}}$. We check whether
  $\|(\hat{s},S)\|_p\le k$.
  We answer ``no'' if for each input string~$s_i \in S$, no length-$n$ string~$\hat{s}$ with $\hd{}(s_i,\hat{s})\le \frac{k}{\sqrt[p]{m}}$ exists which satisfies $\|(\hat{s},S)\|_p\le k$.

  
  It remains to show the running-time bound. Observe that the
  preprocessing for all-zero and all-one columns can be done in
  $O(nm)$ time. After that, for each of the $m$ input strings~$s_i$,
  we search all strings of Hamming distance at most
  $k' \leq \frac{k}{\sqrt[p]{m}}$ to $s_i$, and there are $O(n^{\frac{k}{\sqrt[p]{m}}})$ such strings. For each of them, we compute the
  objective function, which can be accomplished in $O(nm)$ time. As
  already reasoned, after the preprocessing, $n$ is upper-bounded by
  $k^p$. Thus, the overall running time bound is
  $O(nm + nm^2 \cdot n^{\frac{k}{\sqrt[p]{m}}}) = O(nm^2 \cdot
  k^{\frac{p\cdot k}{\sqrt[p]{m}}})$, as claimed.
\end{proof}

\noindent Combining \cref{lem:dp-m} with \cref{lem:XP-n}, we obtain a subexponential 
algorithm with respect to~$k$.

\begin{theorem}\label{thm:subexp}
  For each fixed positive value~$\varepsilon>0$, \pHDC{} can be solved in $O(nm^2 \cdot 2^{k^{\nicefrac{p}{(p+1)}+\varepsilon}})$~time, where $n$ and $m$ denote the length and the number of input strings, and $k$ is the $p$-norm distance bound with $p>1$.
\end{theorem}

\begin{proof}
  Let $I = (S,k)$ be an instance of \pHDC{} with $S=(s_1,\dots,s_m)$ being the input strings of length $n$ and $k$ being the p-norm distance bound.
  As already discussed, to solve our problem we distinguish between two cases, depending on whether $m \le \frac{k^{\nicefrac{p}{(p+1)}}}{\log{k}}$ holds.

  If $m \le \frac{k^{\nicefrac{p}{(p+1)}}}{\log{k}}$, then
  $k^m \le k^{\frac{k^{\nicefrac{p}{(p+1)}}}{\log{k}}}\le
  2^{k^{\nicefrac{p}{(p+1)}}}$. In this case, we use the dynamic
  programming approach given in the proof of \cref{lem:dp-m}, which
  has the desired running
  time~$O(n\cdot k^m) = O(n\cdot 2^{k^{\nicefrac{p}{(p+1)}}})$.
  
  Otherwise, $m > \frac{k^{\nicefrac{p}{(p+1)}}}{\log{k}}$, meaning that
$\frac{p \cdot k \cdot \log{k}}{\sqrt[p]{m}}  < p \cdot k \cdot \log{k}/\sqrt[p]{\frac{k^{\nicefrac{p}{(p+1)}}}{\log{k}}}=p\cdot k^{\nicefrac{p}{(p+1)}}\cdot (\log{k})^{\nicefrac{(p+1)}{p}}$.
  \noindent For each fixed positive~$\varepsilon \in \mathds{R}$ there exists
  $k_0 = k_0(p, \varepsilon) \in \mathds{R}$ such that, for each
  $k \geq k_0$, we have
  $p \cdot (\log{k})^{\nicefrac{(p+1)}{p}} < k^{\varepsilon}$. If $k < k_0$,
  then the algorithm in the proof of~\cref{lem:XP-n} runs in $O(nm^2)$ time.
  Otherwise $k \geq k_0$, which implies
  $\frac{p\cdot k \cdot \log{k}}{\sqrt[p]{m}} <
  k^{\nicefrac{p}{(p+1)}+\varepsilon}$. Thus, the
  algorithm given in the proof of \cref{lem:XP-n} has a running
  time of
  $O(nm^2\cdot k^{\frac{p\cdot k}{\sqrt[p]{m}}}) = O(nm^2\cdot
  2^{\frac{p\cdot k \cdot \log{k}}{\sqrt[p]{m}}}) = O(nm^2\cdot
  2^{k^{\nicefrac{p}{(p+1)}+\varepsilon}})$.

  Altogether we presented an algorithm which has the desired running time bound.
\end{proof}


\subsection{A fixed-parameter algorithm for the number of input strings}\label{sec:FPT-m}

In this section, we show that minimizing the sum of the $p$-distances
is fixed-parameter tractable for the number~$m$ of input strings. The
idea is to formulate our problem as a combinatorial $n$-fold integer
program~(C$n$IP) with $O(2^m)$ variables and $O(m)$ constraints. We
then apply the following simplified result of
\citet{KnoKouMni2017,knop_combinatorial_2017}:

\begin{proposition}[{\cite[Theorem~3]{knop_combinatorial_2017}}]\label[proposition]{thm:nfold}
  Let $E \in \mathds{Z}^{(r + 1) \times t}$ be a matrix such that the
  last row equals~$(1, 1, \ldots, 1) \in \mathds{Z}^t$. Let
  $b \in \mathds{Z}^{r + 1}$, $\ell, u \in \mathds{Z}^{t}$, and let
  $f \colon \mathds{R}^t \to \mathds{R}$ be a separable convex
  function\footnote{A function is separable convex if it is the sum of
    univariate convex functions.}. Then, there is an algorithm that
  solves\footnote{The algorithm correctly reports either a minimizer
    $x \in P$ or that $P$ is infeasible or unbounded.}
  $P := \min\{f(x) \mid Ex = b \wedge \ell \leq x \leq u \wedge x \in
  \mathds{Z}^t\}$ in
  $t^{O(r)}\!\cdot\! \big((1+\|E\|_{\infty})\!\cdot\!
  r\big)^{O(r^2)}\!\cdot\! L + T$ time, where
$L$ is the total bit-length of
$b, \ell, u$, and $f$, and $T$ is the time 
that an algorithm needs to solve the continuous relaxation of~$P$.
\end{proposition}

To get a useful running time bound from \cref{thm:nfold}, we need a
bounded number of variables. To do this, we group columns in the input
strings with the same ``type'' together and introduce an integer
variable for each column type. To this end, given a
set~$S=\{s_1,\ldots,s_m\}$ of length\nobreakdash-$n$ strings, we say
that two columns~$j, j'\in [n]$ have the \myemph{same type} if for each
$i\in [m]$ it holds that $s_i[j]=s_i[j']$. The \myemph{type} of column~$j$ is its equivalence class in the same-type relation. Thus, each
type is represented by a vector in $\{0, 1\}^m$. Let $n'$ denote the
number of different (column) types in~$S$. Then,
$n' \le \min(2^m, n)$. Enumerate the $n'$ column types as
$t_1, \ldots, t_{n'}$. Below we identify a column type with its index
for easier notation. Using this, we can encode the set~$S$ succinctly
by introducing a constant~$e(j)$ for each column type~$j \in [n']$
that denotes the number of columns with type~$j$.

Analogously, given a solution string~$\sol$, we can also encode this
string~$\sol$ via an integer vector~$x\in \{0,1,\ldots,n\}^{n'}$,
where for each type~$j\in [n']$ we define $x[j]$ as the number of ones
in the solution~$\sol$ whose corresponding columns are of type~$j$.
Note that this encodes all essential information in a solution, since
the actual order of the columns is not important (see \cref{ex:types}). Vice versa, each
integer vector in~$x \in \{0,1,\ldots,n\}^{n'}$ satisfying
$0 \leq x[j] \leq e(j)$ for each $j \in [n']$ yields a length-$n$
binary string~$\sol(x)$; it remains to add constraints and a suitable
objective function to ensure that $\sol(x)$ has minimum sum of
$p$-distances to the input strings.

\begin{example}\label{ex:types}
For an illustration, let $S=\{0000,0001,1110\}$.
The set~$S$ has two different column types, represented by $(0, 0, 1)^T$
, call it type~$1$, and $(0, 1, 0)^T$
, call it type~$2$.
There are three columns of type~$1$ and one column of type~$2$.
The solution~$0110$ for $S$ can be encoded by two variables~$x[1]=2$ and $x[2]=0$.
\end{example}

We next introduce $m$ variables~$y \in \{0, 1, \ldots, n\}^{m}$ that
shall be equal to the Hamming distances of each input string~$s_i$,
$i \in [m]$, to the solution~$\sol(x)$ selected by~$x$. To achieve
this, we need a formula specifying the Hamming distance between the
two strings~$s_i$ and $\sol(x)$, and this formula needs to be linear
in~$x$. This can be achieved as follows; for the sake of simplicity,
we let $s_i[j]=1$ if the column of type~$j$ has one in the
$i^{\text{th}}$ row and $s_i[j]=0$ if it has zero in the
$i^{\text{th}}$ row:
$  \hd{}(s_i,\sol(x))  =\! \sum_{j=1}^{n'}\big(s_i[j] \!\cdot\! (e(j)\!-\!x[j])+(1\!-\!s_i[j])\cdot x[j] \big) =\! \sum_{j=1}^{n'}\left( e(j)\cdot s_i[j] + (1\!-\!2s_i[j])\!\cdot\! x[j] \right) =\! w_i+\sum_{j=1}^{n'}x[j]\cdot (1-2s_i[j]),$
where we define $w_i\coloneqq \sum_{j=1}^{n'}e(j)\cdot s_i[j]$, which denotes the number of ones in string~$s_i$. 



We can now formulate an appropriate C$n$IP. The variables are
$x \in \mathds{R}^{n'}$, $y \in \mathds{R}^m$, and a dummy
variable~$z \in \mathds{Z}$. The bounds~$\ell, u$ for the variables
are defined such that 
\begin{inparaenum}[(1)]\item for each $j \in [n']$ it holds that $0 \leq x[j] \leq e(j)$,
  \item for each $i \in [m]$ it holds that $0 \leq y[i] \leq n$,  and
  \item there is virtually no
constraint on~$z$, that is,
$-n' \cdot n + mn \leq z \leq n' \cdot n + mn$.
\end{inparaenum} The objective function
is defined as $f(x, y, z) = \sum_{i = 1}^{n'}y[i]^p$ which is clearly
separable convex over the domain specified by $\ell$ and~$u$. Finally,
the constraint system~$Et = b$, where $t^\top = (x^\top y^\top z)$ is defined such that the first $m$
constraints are $\sum_{j = 1}^{n'}\big(x[j] \cdot (1-2s_i[j])\big) - y[i] = -w_i$,
for each $i \in [m]$, and the last constraint is
$\sum_{j = 1}^{n'}x[j] + \sum_{i = 1}^m y[i] + z = 0$ (note that this
constraint can always be fulfilled by setting~$z$ accordingly).

By the above reasoning, an instance of \pHDC\ is a yes-instance if and only
if $\min\{f(x) \mid Et = b \wedge \ell \leq t \leq u \wedge t \in
\mathds{Z}^{n' + n + 1}\}$ is at most~$k^p$. Plugging in the running
time of \cref{thm:nfold}, and using a polynomial-time algorithm for
the continuous relaxation of the C$n$IP
above~\cite{chubanov_polynomialtime_2016}, we obtain the following.
\begin{theorem}\label{thm:fpt-m}
 \pHDClong\ can be solved in $m^{O(m^2)} \cdot (n\cdot m)^{O(1)}$ time.
\end{theorem}

\subsection{A factor-\boldmath$2$ approximation}\label{sec:2approx}
\appendixsection{sec:2approx}
It is known that by taking an input string that minimizes the largest Hamming distance over all input strings, \textsc{Closest String} can be approximated within factor~$2$.
Indeed, 
using a similar idea, we show that the minimization version of our \pHDC{} problem can also be approximated within factor $2$.
More specifically, we show that an input string which has minimum $p$-norm to all other input strings is a $2$-approximate solution.
    %

\toappendix{
  \noindent Our approximation is based on the following observation.
  \newcommand{\factorlemma}{For each two non-negative integers~$x$ and
  $y$, and for each rational value $p>1$, it holds that
  $(x+y)^{p} \le 2^{p-1}(x^{p}+y^{p})$.}

\begin{lemma}
  \label[lemma]{lem:factor-2}
  \factorlemma%
\end{lemma}

\begin{proof}
  Define $f \colon \mathds{R}^+ \cup \{0\} \to \mathds{R}$ as
  $f(x) = x^p$. Recall that $p > 1$ and thus $f$ is convex. By
  Jensen's inequality~\jensencite\ we thus have
  \[f\left(\frac{x + y}{2}\right) \leq \frac{f(x) + f(y)}{2}.\]
  It follows that
  \[\frac{(x + y)^p}{2^p} \leq \frac{x^p + y^p}{2},\]
  and thus $(x + y)^p \leq 2^{p - 1}(x^p + y^p)$.
\end{proof}
}

\newcommand{\factortwo}{%
  The minimization variant of \pHDC{} can be approximated within factor $2$ in polynomial time.%
}
\begin{proposition}[\appsymb]\label{prop:factor-2}
  \factortwo
\end{proposition}

\appendixproofwithstatement{prop:factor-2}{\factortwo}{
  \begin{proof} 
Let $S$ be a sequence of $m$~input strings 
and let $s_1 \in S$ be some input string that minimizes the $p$-distance to the input strings: $s_1\coloneqq \argmin_{i}\sum_{s \in S}\hd{p}(s_i, s)$.
We show that $s_1$ is a factor-$2$ approximate solution,
i.e.,
$\|(s_1,S)\|_p\le 2\mathsf{OPT}$, where $\mathsf{OPT}$ is the $p$-norm of an optimal solution for $S$.
To this end, let $\sol$ be an optimal solution for $S$ and let
$\mathsf{OPT} = \objectiveroot$. Since $S$ has $m$ input strings, it
has at least one string, denoted as $\hat{s}$, whose $p$-distance to
$\sol$ is at most the arithmetic mean of $\mathsf{OPT}^{p}$:
$\hd{p}(\hat{s}, \sol) \le \frac{1}{m} \sum_{s \in S}^{m}\hd{p}(\sol,
s) = \frac{\mathsf{OPT}^{p}}{m}$. This will be important in
calculating the relation between the $p$-distance of $s_1$ to
$\mathsf{OPT}$ below. Recall that we have selected string~$s_1$ with
minimum sum of $p$-distances. Thus, the following holds:
\begin{align}
  \objectivel{S}{s_1} \le \objectivel{S}{\hat{s}} \le \sum_{s\in S}(\hd{}(\hat{s}, \sol)+\hd{}(s,\sol))^{p}. \label{eq:factor-2}
\end{align}
The last inequality holds because $p>1$ and the Hamming distances fulfill the triangle inequality.

\noindent 
To obtain our desired approximation factor, it suffices to show that $\|(s_1,S)\|^p_p \le (2
\cdot \mathsf{OPT})^p$.
To achieve this, by \eqref{eq:factor-2} and \cref{lem:factor-2}, we derive that
\begin{align*}
  \|(s_1,S)\|^p_p 
  &\stackrel{\eqref{eq:factor-2}}{\le} \sum_{s\in S}(\hd{}(\hat{s}, \sol)+\hd{}(s,\sol))^{p} \stackrel{{\scriptsize \text{Lemma}}~\ref{lem:factor-2}}{\le} \sum_{s\in S}{2^{p-1}}(\hd{p}(\hat{s}, \sol)+\hd{p}(s,\sol))\\
  & = 2^{p-1}(m\cdot \hd{p}(\hat{s},\sol) + \sum_{s\in S} \hd{p}(s,\sol)) \le
    2^{p-1}(m\cdot \frac{\mathsf{OPT}^p}{m} + \mathsf{OPT}^p) = (2 \cdot \mathsf{OPT})^p.
\end{align*}
Note that the second but last inequality holds since $\hat{s}$ was the string that has $p$-distance at most $\frac{\mathsf{OPT}^p}{m}$ to the solution~$\sol$.
\end{proof}
}

\section{Conclusion and Outlook}
We analyzed the complexity of \pHDClong\ for all fixed rational values~$p$ between $p=1$ and $p=\infty$. We believe that the running time bounds established in this paper, of essentially $2^{\Theta(k^{\frac{p}{p + 1}})}\cdot (nm)^{O(1)}$, connect the extreme points $p = 1$ and $p = \infty$ in a very satisfying way. We did not consider the non-norm case of $0 < p < 1$, as it does not fit our clustering motivation very well. But this non-convex case might be of independent interest, and may be the subject of future work.

An interesting generalization of \CloStr{} is \textsc{Closest Substring} in which we seek a string~$\sol$ of a certain specified 
length such that each of the input strings has a substring which is close to~$\sol$ (see, e.g., \citet{MaSun2009}). It would be interesting to see how our results carry over to this and other similar variants. Finally, the fact that the simple 2-factor approximation for \CloStr{} carries over to \pHDC\ may imply that there are similar connections for approximation algorithms. This warrants further investigation into whether \pHDC\ admits a PTAS.

\newpage

\newcommand{\bibremark}[1]{\marginpar{\tiny\bf#1}}



\end{document}